\documentclass[11pt,a4paper]{article}

\usepackage[english]{babel}
\usepackage[T1]{fontenc}
\usepackage{lmodern}

\usepackage[
  a4paper,
  top=1in,
  bottom=1in,
  left=1in,
  right=1in,
  marginparwidth=1.75cm
]{geometry}
\usepackage{setspace}
\usepackage[shortlabels]{enumitem}

\usepackage{amsfonts}
\usepackage{amsmath}
\usepackage{amsthm}
\usepackage{amssymb}
\usepackage{mathtools}
\usepackage{caption}
\usepackage{booktabs}
\usepackage{array}
\usepackage{graphicx}
\usepackage{microtype}
\usepackage{xspace}
\usepackage{xparse}
\usepackage[dvipsnames]{xcolor}

\usepackage{tikz}
\usetikzlibrary{positioning,arrows.meta}

\usepackage{algorithm}
\usepackage{algpseudocode}
\floatname{algorithm}{Algorithm}

\usepackage[colorlinks=true,allcolors=BlueViolet]{hyperref}
\usepackage[capitalise,nameinlink,noabbrev]{cleveref}

\theoremstyle{definition}
\newtheorem{theorem}{Theorem}[section]
\newtheorem{definition}[theorem]{Definition}
\newtheorem{lemma}[theorem]{Lemma}
\newtheorem{corollary}[theorem]{Corollary}
\newtheorem{proposition}[theorem]{Proposition}

\newtheorem*{claim*}{Claim}

\theoremstyle{remark}
\newtheorem{remark}[theorem]{Remark}
\newtheorem{example}[theorem]{Example}

\newtheorem*{remark*}{Remark}
\newtheorem*{note*}{Note}

\theoremstyle{plain}

\newcommand{\GC}{\textup{\textsc{GCircuit}}\xspace}
\newcommand{\PC}{\textup{\textsc{Pure-Circuit}}\xspace}
\newcommand{\PURE}{\textup{\textsf{PURIFY}}\xspace}
\newcommand{\NOT}{\textup{\textsf{NOT}}\xspace}
\newcommand{\OR}{\textup{\textsf{OR}}\xspace}
\newcommand{\NOR}{\textup{\textsf{NOR}}\xspace}
\newcommand{\AND}{\textup{\textsf{AND}}\xspace}
\newcommand{\NAND}{\textup{\textsf{NAND}}\xspace}
\newcommand{\PPAD}{\textup{\textsf{PPAD}}\xspace}

\newcommand{\PP}{\textup{\textsf{P}}\xspace}
\newcommand{\linearFIXP}{\textup{\textsf{Linear-FIXP}}\xspace}
\newcommand{\pcppad}{\ensuremath{\textup{\textsf{PCP-for-PPAD}}}\xspace}

\newcommand{\eps}{\ensuremath{\varepsilon}\xspace}
\newcommand{\bbot}{\ensuremath{\bot}\xspace}
\newcommand{\Sstate}{\mathcal{S}}
\newcommand{\lst}{L^{\star}}
\newcommand{\gvalonly}{\boldsymbol{\mathrm{x}}}
\newcommand{\pvalonly}{\boldsymbol{\mathrm{y}}}

\usepackage{tabularx}

\hypersetup{
  pdftitle={Tight Inapproximability for Generalized Circuits},
  pdfauthor={Panagiotis Kanellopoulos and Themistoklis Melissourgos},
  pdfsubject={},
  pdfkeywords={GCircuit, Pure-Circuit, PPAD, approximation thresholds,
               PCP-for-PPAD}
}

\title{\texorpdfstring{Tight Inapproximability for Generalized Circuits}
{Tight Inapproximability for Generalized Circuits}\\ \text{}}

\author{Panagiotis Kanellopoulos\\
        \small{King’s College London, UK}\\
        \small{\texttt{panagiotis.kanellopoulos@kcl.ac.uk}}
        \and
        Themistoklis Melissourgos\\
        \small{University of Essex, UK}\\
		\small{\texttt{themistoklis.melissourgos@essex.ac.uk}}
        }

\date{\vspace{0.5cm}}

\begin{document}
\maketitle

\begin{abstract}
Generalized circuits are a fundamental tool for establishing
\PPAD-hardness and membership of approximate equilibrium problems. Despite their
importance, the complexity of finding approximate solutions has
remained unresolved for a wide range of approximation parameters.
We settle this question by establishing a tight threshold at
$\eps=1/6$. Specifically, we prove that $\eps$-\GC is
\PPAD-complete for every constant $0<\eps<1/6$, even when
restricted to four gate types, and polynomial-time solvable for
every constant $\eps\ge1/6$. This closes the gap between the
previously known hardness bound of $1/10$ and the polynomial-time
solvability bound of $1/2$.

We further investigate the relaxation in which a fraction
$\delta$ of the gates may be violated. We establish
polynomial-time solvability of $(\eps,\delta)$-\GC for every
constant $\eps>0$ and $\delta>1/5$, and of the fan-out-one formulations of $\delta$-\PC for every constant $\delta>1/8$. These results follow from a decomposition theorem for directed graphs of maximum in-degree two and establish that the \pcppad conjecture cannot be true for the aforementioned values of \eps and $\delta$, unless $\PPAD = \PP$.

To showcase the contribution of our positive results, we apply them to a variety of problems in algorithmic game theory, including public goods games, financial networks, Bayesian, polymatrix, and graphical games, and obtain improved approximation guarantees.
\end{abstract}

\clearpage

\setcounter{tocdepth}{2}
\tableofcontents
\clearpage

\section{Introduction}
\subparagraph{Generalized circuits and PPAD.}
The generalized circuit problem, \GC, was introduced by Chen, Deng and
Teng~\cite{CDT09} as the key intermediate problem in their proof that
computing a Nash equilibrium of a two-player normal-form game is
\PPAD-hard. An instance of \GC consists of variables taking values
in $[0,1]$, together with gates imposing approximate arithmetic,
comparison, and Boolean constraints. The gates may form directed cycles,
so a generalized circuit is a fixed-point problem, and a solution is
guaranteed to exist by Brouwer's fixed-point theorem. In
\eps-\GC, every gate needs to be satisfied only up to
additive error \eps.

Since its introduction, \eps-\GC has become a standard
starting point for proving hardness of approximate equilibrium problems.
For such reductions, the value of \eps matters: hardness for a
larger value of \eps can translate directly into a stronger
inapproximability result for the target problem.
Rubinstein~\cite{Rubinstein18} proved that \eps-\GC is
\PPAD-hard for some sufficiently small constant $\eps>0$.
Deligkas, Fearnley, Hollender and Melissourgos~\cite{PureCircuit}
subsequently introduced \PC, a purely discrete problem over
$\{0,1,\bbot\}$, and proved that it is \PPAD-complete. They then used it to show that \eps-\GC is \PPAD-hard for every $\eps<1/10$. More recent hardness results have sometimes bypassed \GC and reduced directly from \PC, for
example for graphical games~\cite{DFHM23graphical}, financial
networks~\cite{IdKV23}, Fisher
markets~\cite{DFHM24fisher,DFHM26}, and pacing equilibria~\cite{ChenLi26}.

On the algorithmic side, for $\eps\ge1/2$ the problem is trivial:
setting every variable to $1/2$ satisfies every gate. To the best of our
knowledge, no polynomial-time algorithm was previously known for any
constant $\eps<1/2$. This leaves a substantial gap between the
hardness bound $1/10$ and the trivial algorithm at $1/2$.

\medskip
\noindent\emph{For which constant values of \eps is
\eps-\GC \PPAD-hard, and for which is it polynomial-time
solvable?}
\medskip

Besides settling the complexity of a fundamental PPAD problem, an exact
answer determines the strongest approximation parameter that can
possibly be obtained from a reduction starting from
\eps-\GC. Beyond this point, a hardness proof must use a
stronger starting problem, such as \PC, or exploit additional
structure of the target problem.

\subparagraph{Allowing violated gates.}
A different relaxation of \GC allows some gates to be violated
altogether. In the $(\eps,\delta)$-\GC problem, at least a
$(1-\delta)$-fraction of the gates must be \eps-satisfied,
while the remaining gates may be violated arbitrarily. The \pcppad
conjecture\footnote{A recent manuscript generated with an OpenAI model claims a
proof of the quasilinear-size version of this conjecture~\cite{OpenAI2026PCP};
to our knowledge it has not yet been independently verified.} of Babichenko, Papadimitriou and Rubinstein~\cite{BPR16}
states that there exist constants $\eps,\delta>0$ for which
$(\eps,\delta)$-\GC is \PPAD-hard. The conjecture has recently
been used to obtain conditional hardness for Hylland--Zeckhauser
markets~\cite{BLXZ26}, and Fisher markets~\cite{DFHM26}. The latter work additionally showed the converse: if a broad class of Fisher markets with $(1-\delta)$-optimal bundles and $\eps$-over/under-clearing is \PPAD-hard for some constant $\delta$ and \eps, then the \pcppad conjecture is true. They also gave
equivalent formulations of the conjecture in terms of $\delta$-\PC, which, although they preserve the existence for some positive $\delta$, they do not preserve
its numerical value. It is therefore natural to ask, for specific
formulations of \GC and \PC, how large $\delta$ can be before the
problem necessarily becomes polynomial-time solvable.

\subsection{Our contribution}

Our main result identifies the exact approximation threshold of the
standard \eps-\GC problem. Throughout the main body, \GC
refers to the standard definition of the problem, with the gate set of~\cite{CDT09,PureCircuit}, with
multipliers in $[0,1]$ and the usual conventions (see details in
Section~\ref{sec:prelim}).

\begin{theorem}\label{thm:main}
\eps-\GC is \PPAD-complete for every constant
$0<\eps<1/6$, and it can be solved in polynomial time for every
constant $\eps\ge1/6$.
\end{theorem}

Both directions are new: The hardness result improves the previous bound
$\eps<1/10$ of~\cite{PureCircuit}, while the algorithm gives the first
polynomial-time result for a constant $\eps<1/2$ and shows that
the transition occurs exactly at $1/6$.

For completeness, Appendix~\ref{app:conventions} considers an alternative
error-tolerance convention for the gates of \eps-\GC that has been used in~\cite{Rubinstein18}. Away from the precise threshold point, $1/6$, the results hold in both the standard and the alternative variant; the latter affects only
the complexity classification at $\eps=1/6$. In Appendices~\ref{app:alt-four-gates} and \ref{app:hardness}, we also show results for families of \GC with restricted sets of gate types, and a reduced number of gates per simulation gadget, respectively.

Our second set of results concerns violated gates; Section~\ref{sec:robust-overview} gives an overview of the violated-gate
algorithms, with full proofs and parameter-dependent running times in
Appendix~\ref{sec:robust}.

\begin{theorem}\label{thm:robust}
For all constants $\eps>0$ and $\delta>1/5$,
$(\eps,\delta)$-\GC can be solved in polynomial time.
\end{theorem}

We also consider $\delta$-\PC with fan-out one. Deligkas et al.~\cite{DFHM26} showed that the \pcppad conjecture is equivalent to the \PPAD-hardness of this problem for some constant $\delta>0$.

\begin{theorem}\label{thm:pc-robust}
For every constant $\delta>1/8$, all formulations of $\delta$-\PC can be solved in polynomial time on instances in which every node is the input to exactly one gate.
\end{theorem}

Thus, if the \pcppad conjecture holds, any of its fan-out-one $\delta$-\PC formulations cannot have $\delta > 1/8$, unless every problem in $\PPAD$ is polynomial-time solvable.

Finally, in Section~\ref{sec:positive-applications}, to illustrate their impact, we apply our positive results to game-theoretic problems, such as public goods games, financial networks, and Bayesian, polymatrix, and graphical games, achieving improved approximation guarantees in each case.

\subsection{Our techniques}

\subparagraph{Hardness.}
We reduce from the \PPAD-complete problem \PC~\cite{PureCircuit} to \eps-\GC. In our main hardness result, we simulate the original gates using \GC gates, trying to use as few gate-types as possible. We implement the \OR gates directly by the
corresponding Boolean gate of \GC, and simulate \NOT using two other gates with allowed error $\eps < 1/5$. 

The only nontrivial part is the
\PURE gate, which we manage to implement using twelve \GC gates from a set of only four gate types. If the input represents a pure bit, both outputs are forced to represent the same bit in the \PC instance. Otherwise, it forces the output values such that they cannot both represent non-pure bits. The error bounds that make this gadget work correctly dictate $\eps < 1/6$, which allows the whole reduction to go through.

\subparagraph{The algorithm at $\eps=1/6$.}
A natural algorithmic approach is to start from an unresolved assignment
and let gates progressively reveal information about their outputs. The
challenge lies with subtraction: refining its inputs can require its numerical
output first to decrease and later to increase. We therefore propagate
seven \emph{information states}, rather than numerical values directly.

The states form a partial order describing refinement of information,
while a separate map assigns each state one of the values
$1/6,1/3,1/2,5/6$. This numerical map is neither injective nor monotone.
Nevertheless, every non-comparator gate admits a rule that is monotone in
the information order. A comparator behaves differently: once its inputs
first force an output, it commits to that output permanently. A
``no-flip'' lemma proves that subsequent refinements can never force the
opposite value. Since every update strictly increases a state in the
finite information order, the process terminates in polynomial time.

\subparagraph{Violated gates.}
For the violated-gate results on $(\eps,\delta)$-\GC, we work with a dependency digraph whose
arcs point from inputs to outputs. Every vertex therefore has in-degree
at most two. Our structural theorem first deletes at most $n_2/4$
vertices, where $n_2$ is the number of vertices of in-degree two, and
constructs an elimination ordering of the remaining graph with in-width
at most three. Eliminating a vertex means bypassing it by adding
shortcut arcs from its current in-neighbors to its current out-neighbors.

The elimination ordering is then used, for any integer $\lambda\ge1$, to
make every strongly connected component have size at most $\lambda$ using only $O(n/\lambda)$ additional
deletions. Once all inputs entering a component have been fixed, the
component has constant size and can be solved by exhaustive search over a
fixed finite grid. Choosing $\lambda$ as a sufficiently
large constant and balancing this decomposition assignment against two
simple assignments gives the $1/5$ bound of
Theorem~\ref{thm:robust}, as well as a refinement when negation gates are
sparse.

For fan-out-one \PC, we apply the same graph argument to the gate
digraph. Counting inputs and outputs shows that there are exactly as many two-input gates as \PURE gates, so at most half of the gates have in-degree two, leading to the $1/8$ bound.

\subsection{Further related work}

The \pcppad conjecture was introduced by Babichenko, Papadimitriou
and Rubinstein~\cite{BPR16}. Recent work has used it for conditional
hardness results~\cite{BLXZ26,DFHM26}, and Deligkas et
al.~\cite{DFHM26} established equivalent \PC formulations of the
conjecture. Our violated-gate results go in the opposite direction: they
give explicit values of $\delta$ above which particular formulations are
polynomial-time solvable.

As observed in~\cite{SS19}, \eps-\GC is not monotone in
\eps. This matters here because an algorithm at
$\eps=1/6$ does not automatically give an algorithm for every
larger value. Our propagation algorithm is also related in spirit to
classical data-flow analysis and constraint propagation over ordered
information domains~\cite{Kildall73,CousotCousot77,Apt99}.

The graph problem underlying the violated-gate algorithms, deleting
vertices so that every remaining strongly connected component is small,
has been studied as \emph{Bounded Size Strong Component Vertex
Deletion}~\cite{GMM22} and \emph{Directed Component Order
Connectivity}~\cite{BJEGWY22}, together with a generalization~\cite{NRSS24}
and the related Directed Balanced Separator problem~\cite{ALLS25}. These
works give parameterized or approximation algorithms relative to a
prescribed deletion budget. We instead give an explicit worst-case deletion
bound for digraphs of maximum in-degree two, in the spirit of the
undirected fragmentability results of Edwards and Farr~\cite{EF01,EF08,EF12}
and Gupta et al.~\cite{GLLMW19}.

Our proof adapts the deletion-and-reduction charging used for feedback
vertex sets of 2-degenerate graphs~\cite{BorowieckiDS14} and for
planarization~\cite{Meyer26}, i.e., on acyclic orientations of in-degree at most two. Our new
ingredient is a quarter-deletion bound for digraphs of maximum in-degree
two with directed cycles, yielding a directed elimination
ordering~\cite{RoseTarjan78} of in-width at most~$3$, in the sense of
Hunter and Kreutzer~\cite{HK08} applied to the reverse digraph.

\subparagraph{Organization.}
Section~\ref{sec:prelim} gives the definitions.
Section~\ref{sec:hardness} proves the hardness part of
Theorem~\ref{thm:main}, and Section~\ref{sec:algorithm} proves the
algorithmic part. Section~\ref{sec:robust-overview} studies violated
gates, and Section~\ref{sec:positive-applications} gives applications. Appendices~\ref{app:alt-four-gates} and~\ref{app:hardness} give restricted-gate variants of the hardness result, Appendix~\ref{app:conventions} treats the open tolerance convention, Appendix~\ref{app:eps-above-1/6} extends the algorithm to $\eps>1/6$,  Appendix~\ref{sec:robust} contains the proofs for violated gates, and Appendix~\ref{app:positive-applications} the applications. 

\section{Preliminaries}\label{sec:prelim}

\subsection{Generalized circuits}\label{sec:GC}

\begin{definition}[\GC~\cite{CDT09}]\label{def:GC}
A generalized circuit consists of a set $V$ of variables and a set
of gates. Every variable is the output of exactly one gate and may
be used as an input by any number of gates. For $\eps\ge0$, an
assignment $\gvalonly:V\to[0,1]$ is an \eps-\GC solution if it
satisfies every gate constraint below. We write $z=t\pm\eps$
for $t - \eps \leq z \leq t + \eps$.

The arithmetic gates, with inputs $x,y$ and output $z$, impose
\[
\begin{aligned}
    G_\zeta():\;&z=\zeta\pm\eps,  \qquad  \qquad  \qquad  \qquad 
& G_{\times c}(x):\;&z=cx\pm\eps,\\
G_=(x):\;&z=x\pm\eps,
& G_+(x,y):\;&z=\min\{x+y,1\}\pm\eps,\\
G_-(x,y):\;&z=\max\{x-y,0\}\pm\eps.
\end{aligned}
\]
The comparator and Boolean gates impose
\[
\begin{array}{@{}ll@{\qquad}ll@{}}
G_<(x,y):&
\begin{cases}
z\ge1-\eps,&x<y-\eps,\\
z\le\eps,&x>y+\eps,
\end{cases}
&
G_\neg(x):&
\begin{cases}
z\ge1-\eps,&x\le\eps,\\
z\le\eps,&x\ge1-\eps,
\end{cases}
\\[5ex]
G_\vee(x,y):&
\begin{cases}
z\le\eps,&\max\{x,y\}\le\eps,\\
z\ge1-\eps,&\max\{x,y\}\ge1-\eps,
\end{cases}
&
G_\wedge(x,y):&
\begin{cases}
z\le\eps,&\min\{x,y\}\le\eps,\\
z\ge1-\eps,&\min\{x,y\}\ge1-\eps.
\end{cases}
\end{array}
\]
If neither premise of a comparator or Boolean gate holds,
its output is unrestricted. The parameters $\zeta,c\in[0,1]$
are rational numbers given in the input. We write $G_0$ and $G_1$
for the constant gate $G_\zeta$ with $\zeta=0$ and $\zeta=1$, respectively.
\end{definition}

We use the definition above throughout the main body of this work. This is the standard definition in the related literature. In Appendix~\ref{app:conventions} we also analyze an alternative convention used in~\cite{Rubinstein18}.

\subparagraph{An important remark: non-monotonicity in \eps.}
The definition of \eps-\GC has a particular subtlety first noticed by Schuldenzucker and Seuken in~\cite{SS19}. In particular, unlike many approximation problems, it is \emph{not monotone} in \eps: an \eps-solution need not be an
$\eps'$-solution for $\eps'>\eps$. The arithmetic constraints become weaker as \eps grows. For comparator gates, the premises become harder to trigger and the output constraints become weaker. In contrast, although the output constraints of Boolean gates also become weaker, their premises $x\le\eps$ and $x\ge1-\eps$ become easier to trigger. For example, $x=0.18$ and
$z=0.5$ satisfy $z=G_\neg(x)$ when $\eps=1/6$, but not when
$\eps=0.2$, since then $x\le\eps$ forces $z\ge0.8$.
Hence neither hardness nor algorithmic results automatically transfer
between different values of \eps. \cite{SS19} repairs the hardness results which at the time relied on reductions from \GC. Here, we provide standalone hardness proofs by reducing from \PC, bypassing non-monotonicity issues. On the other hand, our approximation algorithm has to carefully handle the non-monotonicity by considering two distinct intervals of \eps's domain, namely $[1/6, 1/3)$ and $[1/3, 1/2)$.

\subsection{\PC}

\begin{definition}[\PC~\cite{PureCircuit}]\label{def:PC}
An instance of the \PC problem is given by a set $V$ of nodes (or \emph{variables}), and a set
$G$ of gate-constraints (or \emph{gates}). Each gate $g \in G$ is of the
form $g = (T,u,v,w)$ where $u,v,w \in V$ are distinct nodes, and $T \in \{\NOT, \OR, \PURE\}$ is the type of the gate, with the following interpretation.
\begin{itemize}
	\item If $T=\NOT$, then $u$ is the input of the gate, and $v$ is its output. Variable $w$ is unused, and for notational simplicity we write $v=\NOT(u)$.
	\item If $T=\OR$, then $u$ and $v$ are the inputs of the gate, and $w$ is its output. We write $w=\OR(u,v)$.
	\item If $T=\PURE$, then $u$ is the input of the gate, and $v$ and $w$ are its outputs. We write $(v,w)=\PURE(u)$.
\end{itemize}
We require that each node is the output of exactly one gate.

A solution to instance $(V,G)$ is an assignment $\pvalonly: V \to \{0,1,\bbot\}$ that satisfies all the gates. These gates' constraints are defined in Figure~\ref{fig:pc}.\footnote{More gates are presented in this figure; these serve equivalent definitions of \PC used in this work.} 
\end{definition}

\begin{figure}[t]
\centering
\small
\resizebox{\linewidth}{!}{%
\begin{tabular}[t]{c||c}
\toprule
$u$ & $v$\\
\midrule
$0$ & $1$\\
$1$ & $0$\\
$\bbot$ & $*$\\
\bottomrule
\multicolumn{2}{c}{\NOT gate}
\end{tabular}
\qquad
\begin{tabular}[t]{cc||c}
\toprule
$u$ & $v$ & $w$\\
\midrule
$1$ & $*$ & $1$\\
$*$ & $1$ & $1$\\
$0$ & $0$ & $0$\\
\multicolumn{2}{c||}{else} & $*$\\
\bottomrule
\multicolumn{3}{c}{\OR gate}
\end{tabular}
\qquad
\begin{tabular}[t]{cc||c}
\toprule
$u$ & $v$ & $w$\\
\midrule
$1$ & $*$ & $0$\\
$*$ & $1$ & $0$\\
$0$ & $0$ & $1$\\
\multicolumn{2}{c||}{else} & $*$\\
\bottomrule
\multicolumn{3}{c}{\NOR gate}
\end{tabular}
\qquad
\begin{tabular}[t]{cc||c}
\toprule
$u$ & $v$ & $w$\\
\midrule
$0$ & $*$ & $0$\\
$*$ & $0$ & $0$\\
$1$ & $1$ & $1$\\
\multicolumn{2}{c||}{else} & $*$\\
\bottomrule
\multicolumn{3}{c}{\AND gate}
\end{tabular}
\qquad
\begin{tabular}[t]{cc||c}
\toprule
$u$ & $v$ & $w$\\
\midrule
$0$ & $*$ & $1$\\
$*$ & $0$ & $1$\\
$1$ & $1$ & $0$\\
\multicolumn{2}{c||}{else} & $*$\\
\bottomrule
\multicolumn{3}{c}{\NAND gate}
\end{tabular}
\qquad
\begin{tabular}[t]{c||cc}
\toprule
$u$ & $v$ & $w$\\
\midrule
$0$ & $0$ & $0$\\
$1$ & $1$ & $1$\\
$\bbot$ & \multicolumn{2}{c}{\shortstack{at least\\one in $\{0,1\}$}}\\
\bottomrule
\multicolumn{3}{c}{\PURE gate}
\end{tabular}%
}
\caption{The gates of \PC. Here $*$ denotes an arbitrary value in
$\{0,1,\bbot\}$.}
\label{fig:pc}
\end{figure}

\begin{theorem}[\cite{PureCircuit}]\label{thm:pc}
\PC with a gate set containing \PURE and a universal set of Boolean gates from Figure~\ref{fig:pc} is \PPAD-complete, even if each node is input to exactly one gate (fan-out is one). In particular, Definition~\ref{def:PC} falls under this category.
\end{theorem}

\section{Hardness for $\eps < 1/6$}
\label{sec:hardness}

Here we prove that for any fixed $0<\eps<1/6$, \eps-\GC is \PPAD-hard even when only four gate types are used, namely, $G_\zeta, G_-, G_<$, and $G_\vee$ (for three more sets consisting of four gate types, see Appendix~\ref{app:alt-four-gates}). Moreover, each \emph{gadget}, i.e., a construction of \GC gates,  uses at most 12 \GC gates, and there is a single $G_1$ gate in the whole construction (for an alternative hardness proof that uses at most 6 gates per gadget, but 6 gate types, see Appendix~\ref{app:hardness}).

We reduce from \PC. The idea is to simulate each type of \PC gate with a gadget, when the tolerance of the latter gates is $\eps < 1/6$. A value of the constructed \GC instance is decoded by
\begin{equation}\label{eq:decode}
[0,\eps]\mapsto0,\qquad
[1-\eps,1]\mapsto1,\qquad
(\eps,1-\eps)\mapsto\bbot.
\end{equation}
We call a \GC node's value \emph{pure} if it decodes to $0$ or $1$.

First, we create one shared node $q=G_1()$, so that $1-\eps \le q \le 1$.
The same node feeds into all gadgets in the reduction that need it. Next, we present the gadgets that simulate the \PC gates.

\subsection{Simulating \NOT and \OR}\label{sec:sim-not-or}

According to Definitions \ref{def:GC} and \ref{def:PC}, gate $G_\vee$ directly simulates \OR.

We next implement \NOT without using $G_\neg$. Given an input $x$, we create the gadget
\begin{equation*}
t=G_-(q,x),
\qquad
y=G_<(x,t).
\end{equation*}
If $x\le\eps$, then
$t-x\ge1-4\eps>\eps$,
so $y\ge1-\eps$. If $x\ge1-\eps$, then
$t\le2\eps$ and
$x-t\ge1-3\eps>\eps$,
so $y\le\eps$. Thus, the gadget simulates \NOT correctly.

\subsection{Simulating \PURE}\label{sec:sim-purify}

We first describe a four-gate block that will be used twice. Given
three \GC nodes $x,y,t$, let block $B$ be defined as $(d,z)=B(x,y,t)$, where
\begin{equation}\label{eq:block}
\begin{aligned}
d&=G_-(x,y),&
p&=G_<(d,y),\\
p'&=G_<(d,t),&
z&=G_\vee(p,p').
\end{aligned}
\end{equation}

The following simple observation is the core of the construction.

\begin{lemma}\label{lem:block}
Suppose that $t=G_-(q,x)$. Then every assignment satisfying the gates
of $B(x,y,t)$ has the following properties:
\begin{enumerate}
    \item $x\le y \text{ } \qquad \qquad \qquad\Longrightarrow\quad d\le\eps \quad\text{and}\quad z\ge1-\eps$, \label{eq:block-first}
    \item $x\ge1-\eps,\quad y\le\eps \quad\Longrightarrow\quad d\ge1-3\eps \quad\text{and}\quad z\le\eps$. \label{eq:block-second}
\end{enumerate}
\end{lemma}

\begin{proof}
Suppose first that $x\le y$. Then $d\le\eps$ from \eqref{eq:block}. Moreover, since $t=G_-(q,x)$ and $q\ge1-\eps$ (recall that $q = G_1()$ is the shared output of the constant gate), we have $x+t\ge1-2\eps$. Therefore, $y+t\ge x+t\ge1-2\eps>4\eps$.
Hence at least one of $y,t$ is strictly larger than $2\eps$.
Since $d\le\eps$, either comparator $G_<(d,y)$ forces $p \geq 1-\eps$ or comparator $G_<(d,t)$ forces $p' \geq 1-\eps$. Hence, $G_\vee(p,p')$ forces $z\ge1-\eps$, proving Property~\ref{eq:block-first} of the lemma. 

Now suppose that $x\ge1-\eps$ and $y\le\eps$.
The subtraction gate $G_-(x,y)$ gives $d\ge x-y-\eps\ge1-3\eps$.
Also, since $q \le 1$, we get $t \le 1-(1-\eps) + \eps = 2\eps$.
Consequently, $d-y\ge1-4\eps>\eps$, and $d-t\ge1-5\eps>\eps$, where the rightmost inequalities hold due to our choice of $\eps < 1/6$.
The outputs of both the comparators in \eqref{eq:block} are forced to at most $\eps$, and hence
$z\le\eps$, proving Property~\ref{eq:block-second} of the lemma.
\end{proof}

We can now construct the \PURE gadget. Given its input $u$, first create
\begin{equation}\label{eq:purify-core}
r=G_-(q,u),\qquad
a=G_<(u,r),\qquad
s=G_-(q,a).
\end{equation}
Then use the block in the two opposite orders:
$(b,w)=B(a,u,s)$, and $(c,z)=B(u,a,r)$, and let $v=G_<(z,c)$.
The outputs of the \PURE gadget are $v,w$.
See Figure~\ref{fig:gadget} for a depiction.

The construction uses three gates in~\eqref{eq:purify-core}, four in
each copy of block $B$, and one final comparator outputting $v$, for a total of
twelve gates. The shared constant node~$q$ is not included in this
count.

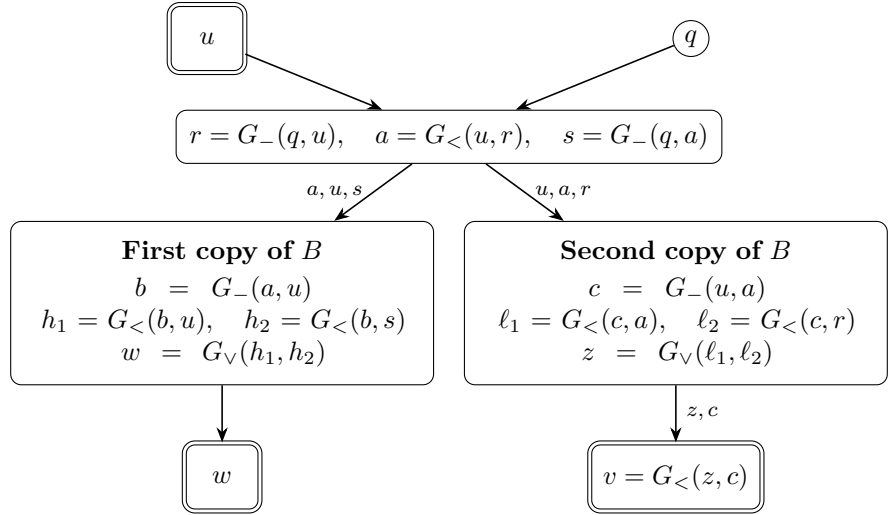
\begin{figure}[t]
\centering
\begin{tikzpicture}[
  >=Stealth,
  gate/.style={
    draw,rounded corners,inner sep=5pt,
    font=\small,align=center
  },
  box/.style={
    draw,rounded corners,inner sep=7pt,
    font=\small,align=center,text width=5.1cm
  },
  var/.style={
    draw,circle,inner sep=2pt,font=\small
  },
  terminal/.style={
    gate,double,double distance=1.1pt,
    minimum width=10mm,minimum height=9mm
  },
  wire/.style={->,semithick}
]

\node[terminal] (u) at (-3.2,0) {$u$};
\node[var]      (q) at (3.2,0) {$q$};

\node[gate] (core) at (0,-1.3)
{$r=G_-(q,u),\quad a=G_<(u,r),\quad s=G_-(q,a)$};

\draw[wire] (u) -- (core);
\draw[wire] (q) -- (core);

\node[box] (left) at (-3.0,-3.5)
{\textbf{First copy of $B$}\\[2pt]
$b=G_-(a,u)$\\
$h_1=G_<(b,u),\quad h_2=G_<(b,s)$\\
$w=G_\vee(h_1,h_2)$};

\node[box] (right) at (3.0,-3.5)
{\textbf{Second copy of $B$}\\[2pt]
$c=G_-(u,a)$\\
$\ell_1=G_<(c,a),\quad \ell_2=G_<(c,r)$\\
$z=G_\vee(\ell_1,\ell_2)$};

\draw[wire]
  (core) -- node[left,font=\scriptsize] {$a,u,s$} (left);
\draw[wire]
  (core) -- node[right,font=\scriptsize] {$u,a,r$} (right);

\node[terminal] (w) at (-3.0,-5.8) {$w$};
\node[terminal] (v) at (3.0,-5.8) {$v=G_<(z,c)$};

\draw[wire] (left) -- (w);
\draw[wire] (right) -- node[right,font=\scriptsize] {$z,c$} (v);

\end{tikzpicture}
\caption{The \PURE gadget with input $u$ and outputs $v,w$,
shown with double borders. The same four-gate block is used twice,
with its first two inputs reversed. Node $q=G_1()$ is shared among
all gadgets that need a constant gate in the reduction.}
\label{fig:gadget}
\end{figure}

\paragraph*{Correctness of the \PURE gadget}

We first show that one of the two outputs is always pure.

\begin{lemma}\label{lem:purity}
Consider the \PURE gadget $(v,w)=\PURE(u)$. Then, 
\begin{enumerate}
    \item if $a \le u$, then $w \ge 1-\eps$, \label{eq:purity-first}
    \item if $a \ge u$, then $v \le \eps$. \label{eq:purity-second}
\end{enumerate}
Hence, for every $u\in[0,1]$, at least one of $v,w$ is pure.
\end{lemma}

\begin{proof}
If $a\le u$, apply Property~\ref{eq:block-first} of Lemma~\ref{lem:block} to $(b,w)=B(a,u,s)$. Since $s=G_-(q,a)$, we obtain $w \ge 1-\eps$. If $a \geq u$, again apply the same property to $(c,z)=B(u,a,r)$. This gives $c\le\eps$, and $z\ge1-\eps$.
Hence $z-c\ge1-2\eps>\eps$, so the final comparator $v=G_<(z,c)$ forces $v\le\eps$.
Thus, in either case, one of the two outputs is pure.
\end{proof}

It remains to check that a pure input is copied to both outputs.

\begin{lemma}\label{lem:pure-inputs}
Consider the \PURE gadget $(v,w)=\PURE(u)$. If $u\le\eps$, then $v,w\le\eps$. If
$u\ge1-\eps$, then $v,w\ge1-\eps$.
\end{lemma}

\begin{proof}
Suppose first that $u\le\eps$. Since $q \geq 1-\eps$ and $r=G_-(q,u)$, we get
$r\ge (1-\eps - \eps) - \eps = 1-3\eps$, and so, $r-u \ge 1-4\eps>\eps$.
Thus, the comparator $a=G_<(u,r)$ forces $a\ge1-\eps$.
In particular, $a \geq 1-\eps > \eps \geq u$, so Property~\ref{eq:purity-second} of
Lemma~\ref{lem:purity} already gives $v \le \eps$.
Moreover, applying Property~\ref{eq:block-second} of Lemma~\ref{lem:block} to
$(b,w)=B(a,u,s)$ gives $w\le\eps$.

Now suppose that $u\ge1-\eps$. Since $q \le 1$, we have $r \le 1 - (1-\eps) + \eps = 2\eps$, and therefore, $u-r\ge1-3\eps>\eps$.
Hence $a\le\eps$.
In particular, $a \le \eps < 1-\eps \le u$, so Property~\ref{eq:purity-first} of Lemma~\ref{lem:purity} already gives $w\ge1-\eps$. Applying Property~\ref{eq:block-second} of Lemma~\ref{lem:block} to $(c,z)=B(u,a,r)$ gives $c\ge1-3\eps$, and $z\le\eps$.
Therefore $c-z\ge1-4\eps>\eps$, and comparator $v=G_<(z,c)$ forces $v\ge1-\eps$.
\end{proof}

Lemmas~\ref{lem:purity} and~\ref{lem:pure-inputs} together show that Figure~\ref{fig:gadget}'s gadget simulates \PURE correctly.

\subsection{The reduction}

\begin{theorem}\label{thm:hardness}
For every constant $0<\eps<1/6$, \eps-\GC is
\PPAD-hard even when restricted to the four gate types
$\{G_\zeta,G_-,G_<,G_\vee\}$.
Moreover, only the constant parameter $\zeta=1$ is required.
\end{theorem}

\begin{proof}
To show the hardness, we
reduce from \PC (Theorem~\ref{thm:pc}). Create the shared node
$q=G_1()$. Replace every \NOT, \OR, and \PURE gate by the gadgets detailed in Section~\ref{sec:sim-not-or} and Section~\ref{sec:sim-purify}, respectively.

Consider any solution of the resulting \GC instance and decode
its values according to~\eqref{eq:decode}. The constructions
of Sections~\ref{sec:sim-not-or} and~\ref{sec:sim-purify} satisfy the \NOT, \OR, and \PURE requirements. Hence the decoded assignment is a solution of the original \PC instance. Finally, each \PC gate is replaced by a constant number of \GC gates, so the
reduction is polynomial.
\end{proof}

\section{A polynomial-time algorithm for
\texorpdfstring{$\eps = 1/6$}{epsilon = 1/6}}
\label{sec:algorithm}

We now prove the algorithmic part of Theorem~\ref{thm:main}. Namely, we present a polynomial-time algorithm that computes a solution to $\tfrac16$-\GC. As discussed in Section~\ref{sec:GC}, there is an inherent non-monotonicity issue with \eps-\GC. This does not allow us to infer that, given a \GC instance, a satisfying assignment for $\eps=1/6$ remains satisfying for any $\eps > 1/6$. Indeed, our $\tfrac16$-\GC algorithm works only for $\eps \in [1/6, 1/3)$, and a slight variation of it has to be used for $\eps \in [1/3, 1/2)$ (notice that for $\eps \geq 1/2$, the all-$1/2$ assignment satisfies every gate constraint). For completeness, we present the $\eps > 1/6$ analysis in Appendix~\ref{app:eps-above-1/6}.

\subsection{Information states}\label{sec:inf-states}

The algorithm uses a set of \emph{(information) states} $\Sstate=\{M,L,C,D,\lst,A,H\}$, and maintains a state assignment $\sigma:V\to\Sstate$.
The states are ordered by the following \emph{information order} $\preceq$:
\begin{equation}\label{eq:order}
M\prec L\prec C\prec D,\qquad
M\prec L\prec C\prec\lst,\qquad
M\prec A\prec H.
\end{equation}
Each state $p\in\Sstate$ also has a \emph{numerical interpretation}, or \emph{value},
$\nu(p)\in[0,1]$:
\begin{equation}\label{eq:nu}
\nu(L)=\nu(\lst)=\tfrac16,\quad
\nu(M)=\nu(C)=\tfrac13,\quad
\nu(D)=\nu(A)=\tfrac12,\quad
\nu(H)=\tfrac56.
\end{equation}
The states' information order and their numerical values are depicted in Figure~\ref{fig:poset}.

Algorithm~\ref{alg:propagation} starts with every variable
in state $M$, which is the least element and represents an unresolved variable, and repeatedly applies the update rules
described below. Non-comparator outputs follow monotone
rules, while comparators commit to an output when first
forced. At termination, the algorithm assigns each variable
$x$ the value $\nu(\sigma(x))$.

At tolerance $1/6$, the only states whose numerical values trigger the
Boolean low and high premises are
\begin{equation}\label{eq:bands}
\{p\in\Sstate:\nu(p)\le\tfrac16\}=\{L,\lst\},
\qquad
\{p\in\Sstate:\nu(p)\ge\tfrac56\}=\{H\}.
\end{equation}
For simplicity, we will say that a node (or variable) of the circuit is \emph{low} if its state is in $\{L,\lst\}$, and it is \emph{high} if its state is $H$.

\begin{figure}[t]
\centering
\begin{minipage}[c]{0.61\textwidth}
\captionsetup[algorithm]{
  justification=raggedright,
  singlelinecheck=false,
  margin=0pt
}
\captionof{algorithm}{The $1/6$-\GC algorithm.}
\label{alg:propagation}
\begin{algorithmic}[1]
\State $\sigma(x)\gets M$ for every variable $x$
\Repeat
  \State $\mathit{changed}\gets\mathsf{false}$
  \For{each gate $g$ with output $z$}
    \If{$g$ is not a comparator}
      \State $p\gets$ state prescribed by $g$'s rule
      \If{$\sigma(z)\ne p$}
        \State $\sigma(z)\gets p$;
               $\mathit{changed}\gets\mathsf{true}$
      \EndIf
    \ElsIf{$\sigma(z)=M$}
      \If{the inputs of $g$ force high}
        \State $\sigma(z)\gets H$;
               $\mathit{changed}\gets\mathsf{true}$
      \ElsIf{the inputs of $g$ force low}
        \State $\sigma(z)\gets L$;
               $\mathit{changed}\gets\mathsf{true}$
      \EndIf
    \EndIf
  \EndFor
\Until{$\mathit{changed}=\mathsf{false}$}
\State \Return $x\mapsto\nu(\sigma(x))$
\end{algorithmic}
\end{minipage}\hfill
\begin{minipage}[c]{0.36\textwidth}
\centering
\begin{tikzpicture}[x=0.85cm,
                     st/.style={font=\small},
                     lab/.style={font=\scriptsize,text=gray}]
\node[st] (M) at (0,0) {$M$};
\node[st] (L) at (-1.3,1) {$L$};
\node[st] (C) at (-1.3,2) {$C$};
\node[st] (D) at (-2.1,3) {$D$};
\node[st] (Ls) at (-0.5,3) {$\lst$};
\node[st] (A) at (1.3,1) {$A$};
\node[st] (H) at (1.3,2) {$H$};
\draw (M)--(L)--(C)--(D);
\draw (C)--(Ls);
\draw (M)--(A)--(H);
\node[lab,right=1pt of M] {$1/3$};
\node[lab,left=1pt of L] {$1/6$};
\node[lab,left=1pt of C] {$1/3$};
\node[lab,left=1pt of D] {$1/2$};
\node[lab,right=1pt of Ls] {$1/6$};
\node[lab,right=1pt of A] {$1/2$};
\node[lab,right=1pt of H] {$5/6$};
\end{tikzpicture}
\captionof{figure}{The seven information states, ordered by
refinement from bottom to top, with their numerical values in gray.}
\label{fig:poset}
\end{minipage}
\end{figure}

\subsection{Update rules for non-comparator gates}\label{sec:localrules}

As a warm-up to the purpose of the information order, notice that the states $L$ and $\lst$ have the same numerical value, but only $L$ can
be refined further. This distinction allows the subtraction rule (i.e., the rule that assigns information states to the nodes of subtraction gates) to decrease its numerical output and later
increase it, while always moving upwards in the information order.

\begin{example}\label{ex:chain}
Let $z=G_-(x,y)$, where $x$ is the output of a comparator. Initially $\sigma(x)=\sigma(y)=\sigma(z)=M$.
The subtraction rule (see below) first changes $\sigma(z)$ to $L$, so the numerical
value of $z$ drops from $1/3$ to $1/6$. If the comparator producing $x$
is then forced high, $\sigma(x)$ becomes $H$ and $\sigma(z)$ becomes
$C$, raising the value of $z$ to $1/3$. If $\sigma(y)$ is later refined
to $L$, then $\sigma(z)$ becomes $D$, raising its value to $1/2$.
Every update moves upwards in the information order even though the
numerical values need not do so.
\end{example}

For every non-comparator gate we define an update rule mapping the states of its
inputs to a state of its output.

\subparagraph{Constants.}
For $G_\zeta$, choose once and for all a state $R(\zeta)$ whose numerical value is at most $1/6$ away from the gate's constant parameter $\zeta$. Formally, choose an arbitrary $R(\zeta)$ satisfying $|\nu(R(\zeta))-\zeta|\le\tfrac16$. One can easily check that such a state always exists since the intervals of radius $1/6$ around $\tfrac16,\tfrac13,\tfrac12,\tfrac56$
cover $[0,1]$. 

\subparagraph{Copy.}
For a copy gate, set $F_=(p)=p$.
Thus, the copied state, and hence its numerical value, is unchanged.

\subparagraph{Addition.}
Set
\[
F_+(p,q)=
\begin{cases}
A,&\nu(p),\nu(q)\le1/3,\\
H,&\text{otherwise.}
\end{cases}
\]

For the remaining rules, it is convenient to define the notion of an \emph{information class}, used to define monotone update rules. We will use the information classes:
\[
\mathcal I_0=\{L,C,D,\lst\},
\qquad
\mathcal I_1=\{A,H\}.
\]

\subparagraph{Subtraction.}
Set
\[
F_-(p,q)=
\begin{cases}
\lst,&p\in\mathcal I_0\text{ or }q\in\mathcal I_1,\\
C,&p=H,\ q=M,\\
D,&p=H,\ q\in\mathcal I_0,\\
L,&\text{otherwise.}
\end{cases}
\]

\subparagraph{Boolean gates.}
First interpret an input state in $\mathcal I_0$ as $0$, one in
$\mathcal I_1$ as $1$, and $M$ as $\bbot$. Apply the corresponding
three-valued \NOT, \OR, or \AND rule from Figure~\ref{fig:pc}.
If that rule forces a bit, prescribe $\lst$ for $0$ or $H$ for $1$;
otherwise prescribe $M$.
Since $\mathcal I_0$ and $\mathcal I_1$ are upward closed,
these rules are monotone in the information order.
Lemma~\ref{lem:valid} verifies that every numerically triggered
Boolean premise is satisfied.

\subparagraph{Multiplication.}
Define
\[
\begin{aligned}
P(p) &=
\begin{cases}
L, & p\in\{M,L,C,\lst\},\\
C, & p\in\{D,A\},\\
D, & p=H,
\end{cases}
\quad
&\text{and set}\quad
F_{\times c}(p) =
\begin{cases}
L, & c\le 2/5,\\
P(p), & 2/5<c<4/5,\\
p, & c\ge 4/5.
\end{cases}
\end{aligned}
\]

The following two lemmas give the properties required of these rules:
numerical validity and monotonicity in the information order.

\begin{lemma}\label{lem:valid}
For every non-comparator gate and every tuple of input states, the
numerical values of the inputs and of the state prescribed for the output
satisfy the $\tfrac16$-\GC constraint of the gate.
\end{lemma}
\begin{proof}
The statement is immediate for constants and copy gates.

\noindent\emph{Addition.}
Every state has numerical value at least $1/6$. If both input values are
at most $1/3$, their truncated sum lies in $[1/3,2/3]$, which
is within $1/6$ of $\nu(A)=1/2$. Otherwise the truncated sum lies in
$[2/3,1]$, which is within $1/6$ of $\nu(H)=5/6$.

\noindent\emph{Subtraction.}
Let $t=\max\{\nu(p)-\nu(q),0\}$.
The output $C$ occurs only when $t=1/2$, and $D$ only when
$t\in\{1/3,1/2,2/3\}$. In every other case the output is $L$ or $\lst$,
and inspection of the four numerical levels shows that
$t\in[0,1/3]$. In all cases the error is at most $1/6$.

\noindent\emph{Boolean gates.}
By~\eqref{eq:bands}, only $L,\lst$ trigger a low premise and only $H$
triggers a high premise. Negation sends the former to $H$ and the latter
to $\lst$. $G_{\vee}$ outputs $\lst$ when both inputs are numerically low and
$H$ whenever either input is numerically high. $G_{\wedge}$ outputs $\lst$
whenever either input is numerically low and $H$ when both inputs are
numerically high. Thus, every forced Boolean output is satisfied.

\noindent\emph{Multiplication.}
If $c\le2/5$, then $c \cdot \nu(p)\in[0,1/3]$,
so the fixed output $L$, of value $1/6$, has error at most $1/6$.

Suppose $2/5<c<4/5$. The four possible numerical input/output pairs
under $P$ are
\[
\tfrac16\mapsto\tfrac16,\qquad
\tfrac13\mapsto\tfrac16,\qquad
\tfrac12\mapsto\tfrac13,\qquad
\tfrac56\mapsto\tfrac12.
\]
For a fixed pair, the error is the absolute value of an affine function
of $c$, and is therefore maximized at an endpoint of $[2/5,4/5]$.
For the four pairs above, the corresponding maximum errors are $\tfrac1{10},\tfrac1{10},\tfrac2{15},\tfrac16$,
respectively. Thus, the error is at most $1/6$ throughout the middle
interval.

Finally, if $c\ge4/5$, the rule copies the input state, and
$|\nu(p)-c\nu(p)|
=(1-c) \cdot \nu(p)
\le\tfrac15\cdot\tfrac56
=\tfrac16$.
\end{proof}

\begin{lemma}\label{lem:mono}
Every non-comparator rule is monotone with respect to $\preceq$ in each
input.
\end{lemma}
\begin{proof}
It suffices to check the six cover relations
\[
M\prec L,\quad
L\prec C,\quad
C\prec D,\quad
C\prec\lst,\quad
M\prec A,\quad
A\prec H.
\]
Constant rules and the copy rule are monotone.

\noindent\emph{Addition.}
The set of states with numerical value at most $1/3$ is downward closed.
Hence refining either input can only leave the output unchanged or move
it from $A$ to $H$. Since $A\prec H$, this is an upward move.

\noindent\emph{Subtraction.}
In the first coordinate, only the covers $M\prec L$ and $A\prec H$ can
change the output. The former changes $L$ to $\lst$ when
$q\notin\mathcal I_1$. The latter changes $L$ to $C$ when $q=M$ and to
$D$ when $q\in\mathcal I_0$. All are upward moves.

In the second coordinate, only $M\prec L$ and $M\prec A$ can change the
output. The former changes $C$ to $D$ when $p=H$. The latter changes
$L$ to $\lst$ when $p\in\{M,A\}$ and changes $C$ to $\lst$ when $p=H$.
Again all changes are upward. Every other cover leaves the output
unchanged.

\noindent\emph{Boolean gates.}
The classes $\mathcal I_0$ and $\mathcal I_1$ are upward closed. Hence
refinement can only replace the unresolved value $\bbot$ by $0$ or $1$.
The three-valued Boolean rules are monotone under this information
refinement, and $M\preceq\lst,H$.

\noindent\emph{Multiplication.}
For $c\le2/5$ the rule is constant, and for $c\ge4/5$ it is copying the input, so these cases are monotone. In the middle interval, the rule prescribes function $P$ which is
monotone: as $p$ gets refined to any $p'$ with $p \preceq p'$, we have $P(p) \preceq P(p')$ (in fact, the only strict changes are $L\prec C$ and $C\prec D$). 
\end{proof}

\subsection{Comparators and the ``no-flip'' lemma}

Monotone output rules do not exist for comparators. For example, both $(L,D)$ and $(D,L)$ are below $(D,D)$ in the information
order, but the first pair forces a high output and the second forces a
low output, which lie on different branches of
Figure~\ref{fig:poset}. That is why comparators require different treatment.

We say that a pair of states $(p,q)$ \emph{forces high} if $\nu(p)<\nu(q)-\tfrac16$, and \emph{forces low} if $\nu(p)>\nu(q)+\tfrac16$.
A pair forces high precisely when either
$q=H$ and $p\ne H$,
or
$p\in\{L,\lst\}$ and $q\in\{D,A\}$.
The forced-low pairs are the reversals.

Rather than updating according to a fixed monotone rule, a comparator
commits to an output when its inputs first force one and never changes
that output afterwards. The following ``no-flip'' lemma shows that this is safe.

\begin{lemma}\label{lem:noflip}
Let $p\preceq p'$ and $q\preceq q'$. If $(p,q)$ forces high, then $(p',q')$ does not force low. Symmetrically, if $(p,q)$ forces low, then $(p',q')$ does not force high.
\end{lemma}

\begin{proof}
Suppose $(p,q)$ forces high. Either $q=H$ and $p\ne H$, or
$p\in\{L,\lst\}$ and $q\in\{D,A\}$. In the first case, $q'=H$, since $H$ is maximal. The refined pair
therefore cannot force low. In the second case, every refinement of $p$ has numerical value at most
$1/2$, while every refinement of $q$ has numerical value at least $1/2$.
Thus
$\nu(p')\le\nu(q')$,
and again the refined pair cannot force low. The other direction is
symmetric.
\end{proof}

Notice that a forced pair may become unforced; for example, $(A,H)$ can refine to
$(H,H)$. This is harmless, since an unforced comparator imposes no
constraint on its output.

\subsection{The propagation algorithm}

Algorithm~\ref{alg:propagation} summarizes the propagation procedure which outputs a $\tfrac16$-\GC solution in polynomial time. We can now prove the main result of this section.

\begin{theorem}\label{thm:algo}
$\tfrac16$-\GC can be solved in polynomial time.
\end{theorem}

\begin{proof}
We first show that every update moves a state strictly upwards in
$\preceq$. A comparator changes only from $M$ to $H$ or from $M$ to $L$.

For a non-comparator gate, suppose its output was last set to
$F(\mathbf p)$ and the current input states are $\mathbf p'$. By
induction, all input states have only moved upwards, so
$\mathbf p\preceq\mathbf p'$.
Lemma~\ref{lem:mono} gives
$F(\mathbf p)\preceq F(\mathbf p')$.
Whenever an update occurs these states are different, so the update is
strictly upward. The first update is also upward because $M$ is the
least state.

The longest chain in $\Sstate$ has three strict steps. Hence every
variable is updated at most three times. A full scan that does not
terminate makes at least one update, so there are $O(|V|)$ scans, each
of which takes $O(|V|)$ time. This gives an $O(|V|^2)$ implementation.
Using a worklist containing gates whose inputs have changed gives linear
running time in the circuit size.

It remains to prove correctness. For a non-comparator gate, at
termination its output state is exactly the state prescribed by its rule
on the final input states, so Lemma~\ref{lem:valid} applies.

Consider a comparator. If its output remains in state $M$, then the final
input pair forces neither output; otherwise the final scan would have
updated it. If the comparator committed earlier, the final input states
are refinements of the states at which it committed. By
Lemma~\ref{lem:noflip}, they cannot force the opposite output. They
therefore either force the same output or force nothing. In the former
case the committed numerical value is $5/6$ or $1/6$, as required.
\end{proof}


\section{Allowing violated gates}\label{sec:robust-overview}

\begin{definition}[$(\eps,\delta)$-\GC]\label{def:delta-gc}
For fixed constants $\eps>0$ and $\delta\in[0,1]$, the
$(\eps,\delta)$-\GC problem takes as input a generalized circuit
of Definition~\ref{def:GC}, with rational constants and multipliers,
and asks for an assignment $\gvalonly:V\to[0,1]$ that
$\eps$-satisfies at least a $(1-\delta)$-fraction of its gates.
\end{definition}

\noindent The problem \eps-\GC is the special case $\delta=0$. A
solution always exists by Brouwer's fixed-point theorem, and
\eps-\GC belongs to \PPAD~\cite{CDT09,PureCircuit}.
Allowing violated gates only weakens the requirements, so
$(\eps,\delta)$-\GC also belongs to \PPAD.

We now prove that $(\eps,\delta)$-\GC is polynomial-time solvable
for all constants $\eps>0$ and $\delta>1/5$, and that fan-out-one
$\delta$-\PC (Definition~\ref{def:delta-pc}) is polynomial-time solvable for every constant
$\delta>1/8$ (Theorems~\ref{thm:robust} and~\ref{thm:pc-robust}).
Both proofs represent the circuit dependencies by a digraph of
maximum in-degree two. We delete a set of vertices so that every
remaining strongly connected component has bounded size, and then
solve the components in topological order. Only the gates
corresponding to deleted vertices may be violated.

\subsection{Decomposing digraphs of in-degree two}
\label{subsec:frag-overview}

All digraphs in this section and Appendix~\ref{sec:robust} are
simple: loops and parallel arcs are suppressed. For a digraph $D$,
we write $V(D)$ for its vertex set and $D-S$ for the subgraph induced
by $V(D)\setminus S$.

\begin{theorem}\label{thm:fragmentation}
Let $D$ be an $n$-vertex simple digraph of maximum in-degree at most two,
and let $n_2$ be the number of vertices of in-degree exactly two. For every
integer $\lambda\ge1$, one can compute in polynomial time a set
$S\subseteq V(D)$ with
\[
  |S|\le \frac{n_2}{4}+\frac{4n}{\lambda+1}
\]
such that every strongly connected component of $D-S$ has at most
$\lambda$ vertices.
\end{theorem}

We use \emph{directed vertex elimination}~\cite{RoseTarjan78}:
eliminating a vertex removes it and adds shortcut arcs
from its current in-neighbors to its current out-neighbors,
suppressing loops and parallel arcs.  An \emph{elimination
ordering} of $D$ is an ordering of all its vertices, which are eliminated
one by one in that order; its \emph{in-width} is the largest in-degree
that a vertex has immediately before it is eliminated.
Hunter and Kreutzer~\cite{HK08} measure the out-degree instead; equivalently, we
apply their notion to the reverse digraph. Elimination should not be
confused with deletion: a deleted vertex is removed without adding
shortcuts and is counted in the deletion set.

The proof, given in Appendices~\ref{subsec:quarter}
and~\ref{subsec:fragment}, has two steps. First, we delete
at most $n_2/4$ vertices to obtain an elimination ordering
of in-width at most three. A decomposition argument then
deletes at most $4n/(\lambda+1)$ additional vertices to
bound the size of every strongly connected component by
$\lambda$. The first step uses a charging argument with
seeds, marks, and repairs to handle directed cycles; the second exploits the
reachability structure of bounded-width elimination
orderings.

A linear number of deletions can be necessary. A first-moment
argument using nearly balanced cuts (Lemma~\ref{lem:fragment-lb})
shows that there is a constant $\gamma\approx0.0819$ such that, for every
$\alpha<\gamma$, there are digraphs of maximum in-degree two for which
more than $\alpha n$ deletions are needed even when the permitted
component size is $o(n)$.

\subsection{Application to \texorpdfstring{\GC}{GCircuit}}
\label{subsec:gc-overview}

\subparagraph{Range of \eps.}
We use the standard convention of Section~\ref{sec:prelim}. For
$\eps\ge1/6$, Corollary~\ref{cor:allbig} gives a polynomial-time
algorithm that satisfies every gate. It therefore remains to consider
$0<\eps<1/6$, which we assume throughout this subsection.

\subparagraph{The dependency digraph.}
A variable is \emph{constant} if it is the output of a constant gate
$G_\zeta$. Its value depends on no other circuit variable, so we fix it to
$\zeta$ and exclude it from the dependency digraph.

\begin{definition}[Dependency digraph]\label{def:dependency}
The \emph{dependencies} of a non-constant variable $v$ are the
distinct non-constant variables other than $v$ that occur as
inputs to the gate producing $v$. The
\emph{dependency digraph} $D$ has the non-constant variables as vertices,
and an arc $uv$ whenever $u$ is a dependency of $v$.
\end{definition}

Arcs point from inputs to outputs. Since every gate has at most
two inputs, $D$ has maximum in-degree at most two.

Let $N_0$ be the number of constant gates, $N_1$ the number of non-constant
gates whose output has at most one dependency, and $N_2$ the number of gates
whose output has two dependencies. Then
\[
  N=N_0+N_1+N_2
\]
is the total number of gates. The dependency digraph $D$ has $N-N_0$
vertices, and the number of its vertices of in-degree two is
$n_2=N_2$. For an assignment $\gvalonly$, let
$\mathrm{viol}(\gvalonly)$ denote the number of gates it violates.

\subparagraph{Two assignments.}
Fix a constant $\tau>0$. We construct two assignments satisfying
\begin{equation}\label{eq:two-bounds}
  \mathrm{viol}(\gvalonly^{\mathrm{decomp}})
    \le \frac{N_2}{4}+\tau N,\qquad
  \mathrm{viol}(\gvalonly^{\mathrm{const}})
    \le N_1.
\end{equation}
These assignments are useful in complementary regimes:
$\gvalonly^{\mathrm{decomp}}$ when $N_2$ is small and
$\gvalonly^{\mathrm{const}}$ when $N_1$ is small.

To construct $\gvalonly^{\mathrm{decomp}}$, we apply
Theorem~\ref{thm:fragmentation} to $D$ with
$\lambda=\lceil4/\tau\rceil$. We set the deleted variables to zero and
process the remaining strongly connected components in topological order.
Every boundary input of the current component has already been assigned:
it is either a constant variable, a deleted variable, or belongs to an
earlier component. Lemma~\ref{lem:small-component} guarantees that, for
these fixed boundary inputs, the component has a satisfying assignment on
a fixed finite grid. Since the component has constant size, such an
assignment can be found by exhaustive search. Hence only the gates
producing deleted variables may fail.

The assignment $\gvalonly^{\mathrm{const}}$ gives every constant variable
its prescribed value and sets every other variable to zero.
Constant gates are satisfied. Every gate counted in $N_2$ has both
inputs and its output equal to zero, which satisfies each possible
gate type $G_+$, $G_-$, $G_<$, $G_\vee$, or $G_\wedge$.
Thus, only the $N_1$ remaining gates may fail.
Lemma~\ref{lem:dep-cert} proves the second bound in
\eqref{eq:two-bounds}.

Let $\gvalonly^\star$ be the better of
$\gvalonly^{\mathrm{decomp}}$ and $\gvalonly^{\mathrm{const}}$. From
\eqref{eq:two-bounds},
\[
  5\,\mathrm{viol}(\gvalonly^\star)
  \le N_2+4\tau N+N_1
  \le (1+4\tau)N.
\]
Thus, by choosing $\tau>0$ sufficiently small, we obtain every violation
tolerance strictly above $1/5$, as stated in
Theorem~\ref{thm:robust}. Appendix~\ref{subsec:balancing} gives the
precise parameter choices and a refinement that also takes the number of
negation gates into account.

\subsection{Fan-out-one \texorpdfstring{\PC}{Pure-Circuit}}
\label{subsec:pc-overview}

Deligkas et al.~\cite[Corollary~2.3]{PureCircuit} show that \PC has
several polynomially equivalent versions that use the \PURE gate
together with a universal set of Boolean gates and keep fan-out one:
each version reduces to another by simulating its Boolean gates with
constant-size gadgets. Deligkas et al.~\cite{DFHM26} establish the
equivalence between the existence of constants $\eps,\delta_{\mathrm{G}}>0$
giving PPAD-hardness of $(\eps,\delta_{\mathrm{G}})$-\GC and the
existence of a constant $\delta_{\mathrm{P}}>0$ giving PPAD-hardness
of a fan-out-one $\delta_{\mathrm{P}}$-\PC formulation. Constant-size
Boolean simulations give equivalent existential hardness conjectures
for the gate bases considered here. These reductions need not preserve
the numerical violation parameter: in particular,
$\delta_{\mathrm{P}}$ and $\delta_{\mathrm{G}}$ need not be equal.

\begin{definition}[Fan-out-one $\delta$-\PC]
\label{def:delta-pc}
An instance is a pair $(V,G)$ consisting of a set $V$ of nodes and a set $G$ of gates. Each gate is of the form $g=(T,u,v,w)$,
where $u,v,w\in V$ are distinct and
$T$ can be taken from a set that contains \PURE together with some universal set of Boolean gates from $\{\NOT, \OR, \NOR, \AND, \NAND\}$. Figure~\ref{fig:pc} depicts the input/output constraints of these gates. 
Every node is the output of exactly one gate and also input to exactly one gate (fan-out is one). 

For $\delta\in[0,1]$, a solution of
$\delta$-\PC is an assignment $\pvalonly:V\to\{0,1,\bbot\}$
that satisfies at least a $(1-\delta)$-fraction of the gates.
\end{definition}

The \emph{gate digraph} of a \PC instance with fan-out one is defined as follows. Since each output of a gate is input to exactly one other gate, we represent each gate by a vertex.  Also, every node gives an arc from the gate that \emph{produces} it to the gate that \emph{consumes} it, i.e., uses it as an input. If both outputs of a \PURE gate are inputs of the same gate with fan-in two, we keep only one of the resulting parallel arcs. A \PURE gate has in-degree at most one, while all other gates have in-degree at most two.

For any of the $\delta$-\PC versions of Definition~\ref{def:delta-pc}, let $m_P$, $m_N$, $m_R$, be the numbers of \PURE, \NOT, and the rest of the gates, respectively; notice that in the latter group gates have two inputs and one output. Let $m=m_P+m_N+m_R$ be the total number of gates. Since every arc has a unique producer and a unique consumer, double counting of the arcs (considering the graph before suppressing parallel arcs) by enumerating gates first as producers and then as consumers gives
$2m_P + m_N + m_R = m_P + m_N + 2m_R$, and hence $m_P = m_R \leq \frac m2$.
Thus, at most $m/2$ gates are \PURE gates, and the same holds for gates of in-degree two. Applying
Theorem~\ref{thm:fragmentation} to the gate digraph therefore yields
a set of at most
\[
  \frac{m}{8}+\frac{4m}{\lambda+1}
\]
gates whose deletion leaves every strongly connected component
with at most $\lambda$ gates.

We assign arbitrary values to the outputs of the deleted gates
and process the remaining components in topological order.
Once the values entering a component have been fixed, all gates
in the component can be satisfied simultaneously
(Lemma~\ref{lem:pc-boundary-totality}); since the component has
constant size, a satisfying assignment can be found by exhaustive
search. Thus, only the deleted gates may fail. Choosing $\lambda$
sufficiently large gives every violation tolerance $\delta>1/8$,
proving Theorem~\ref{thm:pc-robust}.

\section{Applications to other problems}\label{sec:positive-applications}

Our positive results yield polynomial-time algorithms for several equilibrium problems. The applications use either the numerical $\tfrac16$-\GC guarantee, the existence of any \eps-\GC algorithm for fixed $\eps<1/2$, direct state propagation, or our decomposition theorem.
Table~\ref{tab:applications-summary} summarizes the guarantees, while
definitions, assumptions, and proofs are deferred to
Appendix~\ref{app:positive-applications}.

\begin{table}[ht]
\centering
\small
\caption{Selected applications of our algorithms.
Errors use the model-specific approximation notions
defined in Appendix~\ref{app:positive-applications}.
The ``Previous'' column reports published guarantees
or elementary baselines, as indicated.}
\label{tab:applications-summary}
\renewcommand{\arraystretch}{1.5}
\begin{tabularx}{\linewidth}{@{}
  >{\raggedright\arraybackslash}X
  >{\raggedright\arraybackslash}p{0.14\linewidth}
  >{\raggedright\arraybackslash}p{0.17\linewidth}
  >{\raggedright\arraybackslash}p{0.23\linewidth}
  @{}}
\toprule
Model and restrictions & Previous & Our guarantee & Method \\
\midrule
Divisible public goods; strict preference games;
maximum in-degree two (two players above self)
& $1/2$ (elementary)
& $7/19$; $1/5$
& Numerical $1/6$-\GC;
direct state propagation
\\\hline

Restricted CDS networks; at most two held unit CDSs
per core bank
& $1/2$ (trivial~\cite{AppDHK25})
& $7/19$-residual;
$1/5$-approximate clearing
& Numerical $1/6$-\GC;
direct state propagation
\\\hline

Binary best-shot public goods;
maximum in-degree two, price $1/2$
& $1/2$-WSNE~\cite{AppDDH24};
$1/4$-NE (elementary)
& $\approx0.119$-WSNE;
$\approx0.058$-NE
& Any fixed $\eps$-\GC solver
with $\eps<1/2$ (WSNE);
direct state propagation (NE)
\\\hline

Normalized two-action polymatrix games;
two-action Bayesian games with two players
or pairwise-additive payoffs
& $1/5$-NE~\cite{AppEASSS23};
$1/5$-interim Bayesian NE$^\dagger$
& $0.191$-NE;
$0.191$-interim Bayesian NE
& Direct state propagation;
type-agent reduction
\\\hline

Threshold games; arbitrary in-degree
& $1/6$~\cite{PureCircuit} (tight)
& $1/6$ (alternative algorithm)
& Direct state propagation
\\\hline

Constant-action polymatrix and graphical games;
maximum in-degree two
& ---
& $3/4-o(1)$ of players satisfy
equilibrium conditions
& Graph decomposition
\\
\bottomrule
\end{tabularx}

\vspace{2pt}
\parbox{\linewidth}{\footnotesize
$^\dagger$ The $1/5$ interim Bayesian NE baseline
follows by applying the type-agent representation
to the earlier $1/5$ polymatrix NE algorithm;
it is not claimed as a separately obtained
Bayesian guarantee.
}
\end{table}

\subparagraph*{Divisible public goods games, preference games with strict preferences, and financial networks.}
A direct reduction to $\tfrac16$-\GC gives a $7/19$ approximation for
 divisible public goods with maximum
in-degree two, preference games with strict preferences
and at most two players
ranked above self, and residual clearing in restricted CDS networks.
The latter have zero-endowment core banks, funded writers, fixed outside
debt, and at most two held unit CDSs per core bank. Direct use of
Algorithm~\ref{alg:propagation} improves all three guarantees to $1/5$
and additionally ensures $\eps$-approximate clearing. These improve elementary
$1/2$ baselines for the game models~\cite{AppPP21,AppKPRST13} and the
trivial clearing bound noted in~\cite{AppDHK25}. Without degree
restrictions, propagation gives $1/3$ approximations for divisible
public goods and preference games with strict preferences. See Appendix~\ref{app:positive-goods-preferences-clearing}.

\subparagraph*{Binary best-shot public goods games.}
Any \eps-\GC algorithm for fixed $\eps<1/2$ improves the general well-supported
Nash equilibrium guarantee of~\cite{AppDDH24} for every fixed price and finite
in-degree bound, yielding exact equilibria at maximum in-degree one.
At price $1/2$ and maximum in-degree two, the error decreases from
$1/2$ to approximately $0.119$. Direct state propagation gives a
stronger ordinary Nash equilibrium guarantee: approximately $0.058$ in this
degree-two setting, improving the elementary $1/4$ baseline, and at
most $4/27$ for arbitrary in-degree and price. See Appendix~\ref{app:positive-binary}.

\subparagraph*{Polymatrix and Bayesian games.}
Direct state propagation computes $0.191$-Nash equilibria in two-action
polymatrix games, improving the $1/5$ guarantee of~\cite{AppEASSS23}.\footnote{The algorithm that computes $1/3$-well-supported Nash equilibria in~\cite{PureCircuit} has an easy adaptation to compute $1/5$-Nash equilibria~\cite{AppEASSS23}.} Type-agent reductions transfer both the earlier
bounds and our improvement to interim equilibria in two-player Bayesian
games and pairwise-additive multiplayer Bayesian games with two actions
per type. See Appendix~\ref{app:positive-polymatrix-bayesian}.

\subparagraph*{Threshold games.}
The same propagation framework gives an alternative polynomial-time
algorithm attaining the tight $1/6$ approximation of~\cite{PureCircuit},
with no in-degree restriction. See Appendix~\ref{app:positive-threshold}.

\subparagraph*{Violating equilibrium conditions for a fraction of players.}
The decomposition results of Section~\ref{sec:robust-overview} apply to
constant-action polymatrix and graphical games with maximum
in-degree two. They ensure exact best responses for every supported
action in polymatrix games, and approximate best responses to any fixed
positive tolerance in graphical games, at all but
$n_2/4+O(n/\log n)$ players, where $n_2$ counts players of in-degree two.
The satisfied fraction is therefore at least $3/4-o(1)$.
To the best of our knowledge, in the current literature there is no other such guarantee for these succinct game representations. See Appendix~\ref{app:positive-fragment}.

\section{Conclusions}
We have established the exact approximation threshold for the
standard generalized circuit problem. We have shown that the problem is \PPAD-complete for every constant $0<\eps<1/6$, and provided a polynomial-time algorithm for every  $\eps\ge1/6$.

We have also established polynomial-time algorithms for
$(\eps,\delta)$-\GC when $\delta>1/5$ and for the fan-out-one
formulations of $\delta$-\PC when $\delta>1/8$. Neither bound is known to be tight; nevertheless, we narrow the range of values for which the \pcppad conjecture can be true. Establishing PPAD-hardness for an explicit
positive constant $\delta$ in the \GC or \PC formulation,
and determining the boundary between tractable and intractable parameter
regimes, remain major open problems.

\paragraph*{Use of generative AI tools.}
ChatGPT 6 (OpenAI) and Claude 5.5 (Anthropic) assisted with proof
exploration, checking mathematical arguments, simplifying
and reorganizing proofs, reference searches, and \LaTeX{}
editing. AI-assisted Python computations verified finite
gate tables, monotonicity and no-flip properties, and
stress-tested the graph decomposition algorithms and their
applications. The authors independently checked all
mathematical claims and proofs and take full responsibility
for the paper.

\newpage

\bibliographystyle{plainurl}
\bibliography{references}

\clearpage
\newpage

\appendix

\section{Other four-type bases that maintain inapproximability for $\eps < 1/6$}
\label{app:alt-four-gates}


The same hardness result as that of Theorem~\ref{thm:hardness} holds for the family of \GC with available gates
\[
\{G_{\neg},G_-,G_<,G_\vee\}.
\]
Indeed, the proof only uses $G_\zeta$ to obtain the shared node
$q\in[1-\eps,1]$. Without constant gates, introduce instead
\[
z_0=G_-(q,q),
\qquad
q=G_{\neg}(z_0).
\]
The subtraction gate gives $z_0\le\eps$, and hence the negation
gate forces $q\ge1-\eps$. The remainder of the construction is
unchanged.
 
Also, in our set of gate types we can use $G_\wedge$ instead of $G_\vee$: simply use De Morgan's law $x \vee y = \neg (\neg x \wedge \neg y)$, and the \NOT gadget from our reduction.\footnote{However, notice that this implementation might increase the necessary number of gates participating in each gadget.} Therefore,
for every fixed $0<\eps<1/6$, \eps-\GC is
\PPAD-hard for each of the four bases
\[
\{G_\zeta,G_-,G_<,G_\vee\},\qquad
\{G_\zeta,G_-,G_<,G_\wedge\},
\]
\[
\{G_{\neg},G_-,G_<,G_\vee\},\qquad
\{G_{\neg},G_-,G_<,G_\wedge\}.
\]

\section{An alternative hardness proof with fewer gates per gadget}
\label{app:hardness}

We give an alternative reduction to that of Section~\ref{sec:hardness}, which used 4 gate types and each of its gadgets was constructed with at most 12 gates. The reduction we show here uses more gate types, i.e., the following 6: $G_{\zeta}, G_-, G_+, G_<, G_{\neg}$, and $ G_{\vee}$, but fewer gates per gadget, namely, 6. These changes come solely from a different \PURE gadget construction that we present in this section.

We reduce \PC to \eps-\GC, with the same decoding as in the reduction of Section~\ref{sec:hardness} (see \eqref{eq:decode}).
Under this decoding, $G_\neg$ and $G_\vee$ implement \NOT and \OR directly. It
therefore remains only to implement \PURE.

First, similarly to the aforementioned construction, we create node $q=G_1()$, which is shared among all \PURE gadgets that need it. Therefore, $1-\eps \le q \le 1$.

\subsection{Alternative simulation of \PURE}\label{app:sim-alt-pure}

Given the input $u$ of a \PURE gate, introduce 
\begin{equation}\label{eq:macro1}
r=G_-(q,u),\quad
s=G_+(u,u),\quad
a=G_<(u,r),\quad
d=G_+(r,a),
\end{equation}
and let the two outputs be
\begin{equation}\label{eq:macro2}
v=G_<(a,s),\qquad
w=G_<(d,u).
\end{equation}
See Figure~\ref{fig:gadget-alt} for a depiction. 


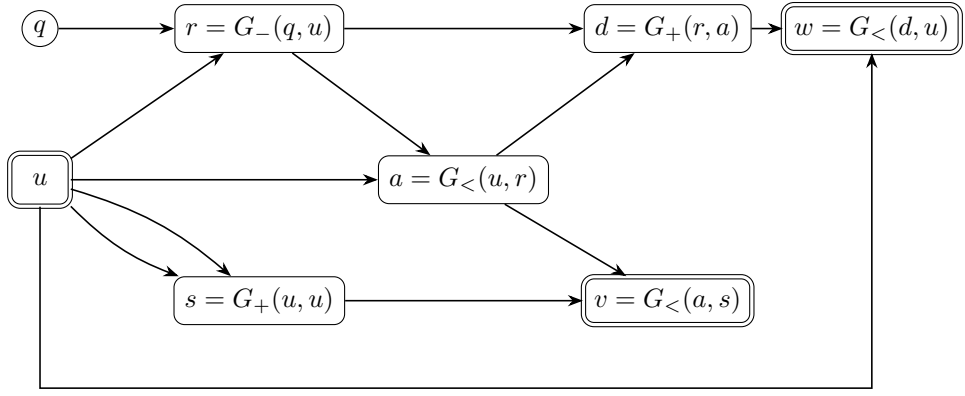
\begin{figure}[ht]
\centering
\begin{tikzpicture}[
  >=Stealth,
  gate/.style={
    draw,rounded corners,inner sep=4pt,font=\small
  },
  outgate/.style={
    gate,double,double distance=1.2pt
  },
  ingate/.style={
    outgate,minimum width=8mm,minimum height=7mm
  },
  wire/.style={->,semithick}
]
\node[ingate] (u) at (0,0) {$u$};
\node[draw,circle,inner sep=2pt,font=\small] (q) at (0,2) {$q$};

\node[gate] (r) at (2.9,2) {$r=G_-(q,u)$};
\node[gate] (s) at (2.9,-1.6) {$s=G_+(u,u)$};

\node[gate] (a) at (5.6,0) {$a=G_<(u,r)$};

\node[gate] (d) at (8.3,2) {$d=G_+(r,a)$};
\node[outgate] (v) at (8.3,-1.6) {$v=G_<(a,s)$};

\node[outgate] (w) at (11.0,2) {$w=G_<(d,u)$};

\draw[wire] (q) -- (r);
\draw[wire] (u) -- (r);
\draw[wire] (u) to[bend right=12] (s);
\draw[wire] (u) to[bend left=12] (s);
\draw[wire] (u) -- (a);
\draw[wire] (r) -- (a);
\draw[wire] (r) -- (d);
\draw[wire] (a) -- (d);
\draw[wire] (a) -- (v);
\draw[wire] (s) -- (v);
\draw[wire] (d) -- (w);
\draw[wire] (u.south) -- ++(0,-2.4) -| (w.south);
\end{tikzpicture}
\caption{The six-gate \PURE gadget with input $u$ and outputs
$v$ and $w$. Node $q=G_1()$ is not created specifically for this gadget; it is created independently and shared among all gadgets that need it.}
\label{fig:gadget-alt}
\end{figure}

\subsection{Correctness of the \PURE gadget}

The constant gate $q = G_1()$ implies that $q\in[1-\eps,1]$. Hence the subtraction
gate satisfies
\begin{equation}\label{eq:r}
1-u-2\eps\le r\le1-u+\eps.
\end{equation}
Indeed, $r\le\max\{q-u,0\}+\eps\le1-u+\eps$, while $r\ge q-u-\eps\ge1-u-2\eps$.

We first show that pure inputs are copied to both outputs.

\begin{lemma}
\label{lem:zero}
Consider the \PURE gadget $(v,w)=\PURE(u)$. If $u\le\eps$, then $v\le\eps$ and $w\le\eps$.
\end{lemma}

\begin{proof}
By~\eqref{eq:r}, we have $r\ge1-3\eps$, so $r-u\ge1-4\eps>\eps$. The comparator $a=G_<(u,r)$ therefore forces $a\ge1-\eps$.
Since $s\le\min\{2u,1\}+\eps\le3\eps$, we have $a-s\ge1-4\eps>\eps$, which forces the output $v\le\eps$ due to the comparator $v=G_<(a,s)$. For the second output,
\[
r+a\ge2-4\eps>1,\qquad
d\ge1-\eps,\qquad
d-u\ge1-2\eps>\eps.
\]
The comparator $w=G_<(d,u)$ therefore forces output $w\le\eps$.
\end{proof}

\begin{lemma}\label{lem:one}
Consider the \PURE gadget $(v,w)=\PURE(u)$. If $u\ge1-\eps$, then $v\ge1-\eps$ and $w\ge1-\eps$.
\end{lemma}

\begin{proof}
By~\eqref{eq:r}, we have  $r\le2\eps$, so $u-r\ge1-3\eps>\eps$. Thus, the comparator $a=G_<(u,r)$ forces $a\le\eps$. Since $2u\ge1$, the addition gate $s=G_+(u,u)$ gives $s\ge1-\eps$, and therefore
$s-a\ge1-2\eps>\eps$.
Hence, $v=G_<(a,s)$ forces output $v\ge1-\eps$. For the second output,
\[
d\le r+a+\eps\le4\eps,
\qquad
u-d\ge1-5\eps>\eps.
\]
The final inequality is precisely where $\eps<1/6$ is needed.
Thus, the comparator $w=G_<(d,u)$ forces output $w\ge1-\eps$.
\end{proof}

We next show that the gadget always produces at least one pure output.

\begin{lemma}\label{lem:purity-alt}
Consider the \PURE gadget $(v,w)=\PURE(u)$. For every $u\in[0,1]$ and every assignment satisfying
the gadget, either $v\ge1-\eps$ or $w\le\eps$.
\end{lemma}

\begin{proof}
Assume $v<1-\eps$. We show that $w$ is pure low. Since the comparator $v=G_<(a,s)$ is not forced high,
\begin{equation}\label{eq:star}
a\ge s-\eps
  \ge\min\{2u,1\}-2\eps.
\end{equation}

For the sake of contradiction, assume $a \leq \eps$. Then \eqref{eq:star} implies that $\min\{2u,1\}\le3\eps<1$, and hence $u\le3\eps/2$. But then, by the lower bound of $r$ from \eqref{eq:r} we get
$r-u\ge1-2u-2\eps \ge1-5\eps >\eps$,
so the comparator $a=G_<(u,r)$ would force
$a\ge1-\eps>\eps$, a contradiction.

Since $a>\eps$, that comparator cannot be forced low. Hence $r\ge u-\eps$.
Together with the upper bound $r\le1-u+\eps$ from
\eqref{eq:r}, this also gives $u\le\tfrac12+\eps$.

Using~\eqref{eq:star} and the two lower bounds on $r$, namely $r \geq 1 - u - 2\eps$ (from \eqref{eq:r}) and $r \geq u - \eps$, we get
\begin{align*}
r+a
&\ge\min\{r+2u,r+1\}-2\eps\\
&\ge\min\{1+u-2\eps,1+u-\eps\}
      -2\eps\\
&=1+u-4\eps
 >u+2\eps,
\end{align*}
where the last inequality is due to the fact that $\eps < 1/6$.
Moreover, $u+2\eps\le\tfrac12+3\eps<1$.
Thus, both arguments of $\min\{r+a,1\}$ exceed
$u+2\eps$, and therefore gate $d=G_+(r,a)$ gives $d\ge\min\{r+a,1\}-\eps>u+\eps$. Finally, comparator $w=G_<(d,u)$ forces $w\le\eps$.
\end{proof}

\subsection{The reduction}

\begin{theorem}\label{thm:hardness-alt}
For every constant $0<\eps<1/6$, \eps-\GC is
\PPAD-hard even when restricted to the six gate types
$\{G_{\zeta}, G_-, G_+, G_<, G_{\neg}, G_{\vee}\}$.
Moreover, only the constant parameter $\zeta=1$ is required.
\end{theorem}

\begin{proof}
We reduce from \PC (Theorem~\ref{thm:pc}). Create one \GC variable
for every \PC variable. Implement every \NOT gate by $G_\neg$, every
\OR gate by $G_\vee$, and every \PURE gate by the gadget construction described in Appendix~\ref{app:sim-alt-pure}.

Consider any solution of the resulting \GC instance and decode all
values according to~\eqref{eq:decode}. The gates $G_\neg$ and $G_\vee$
satisfy exactly the requirements of \NOT and \OR on pure inputs. For
a \PURE gate, Lemmas~\ref{lem:zero} and~\ref{lem:one} show that a
pure input is copied to both outputs, while Lemma~\ref{lem:purity-alt} shows
that at least one output is always pure. Thus, the decoded assignment is a
solution of the original \PC instance.
\end{proof}

\section{An alternative convention}
\label{app:conventions}

The main body uses the standard definition of \eps-\GC (Definition~\ref{def:GC}) found in the great majority of the literature, which we call \emph{closed tolerance}. Here, for completeness, we consider another \eps-\GC convention, called \emph{open tolerance}, in which all gates' input and output constraints become strict inequalities. The latter definition has been used in~\cite{Rubinstein18}. 

We show that by considering open instead of closed tolerance, the complexity exactly at the threshold $\eps=1/6$ crosses sides, i.e., \eps-\GC becomes \PPAD-hard.

\begin{theorem}\label{thm:conventions}
For either tolerance convention,
\eps-\GC is \PPAD-complete for every constant
$0<\eps<1/6$ and polynomial-time solvable for every constant
$\eps>1/6$. At $\eps=1/6$, the problem is polynomial-time solvable under closed tolerance, and \PPAD-complete under open tolerance.
\end{theorem}

The notation and terminology of this section follow those of Definition~\ref{def:GC} of \eps-\GC. 
We first extend the algorithm of Section~\ref{sec:algorithm} to the
open tolerance convention.

\begin{proposition}
\label{prop:alternative-algorithms}
Under open tolerance, \eps-\GC can be solved in polynomial time for every $\eps>1/6$.
\end{proposition}

\begin{proof}
Run Algorithm~\ref{alg:propagation} at tolerance $1/6$,
and let $\sigma$ be its final state assignment.

For $1/6<\eps\le1/3$, use the original numerical
interpretation $\nu$. Every arithmetic error is at most
$1/6<\eps$. The only values satisfying an open Boolean-low
or Boolean-high premise are $1/6$ and $5/6$, respectively.
The Boolean rules give the corresponding outputs, which
satisfy the open output constraints. Moreover, every comparator
forced at tolerance $\eps$ was already forced in the same
direction at tolerance $1/6$, so its committed output
satisfies the required open output constraint.

For $1/3<\eps\le1/2$, use the interpretation $\mu$
of~\eqref{eq:mu}. By Lemma~\ref{lem:readout}, every
arithmetic error is at most $1/3<\eps$. Only the values
$1/3$ and $2/3$ trigger open Boolean-low and Boolean-high
premises, respectively, and the Boolean rules give the
corresponding outputs. No comparator is forced, since
the largest difference between two values under $\mu$
is $1/3<\eps$.

Finally, for $\eps>1/2$, the all-$1/2$ assignment works.
Every arithmetic error is at most $1/2<\eps$, no comparator
is forced, and $1-\eps<1/2<\eps$ satisfies every possible
open Boolean output constraint.
\end{proof}

We next show \PPAD membership of \eps-\GC for constant \eps. For the closed tolerance (standard) convention, this is known from \cite{CDT09}. We give here a standalone proof of membership in \PPAD for both the closed (standard) and the open tolerance convention.

\begin{lemma}\label{lem:membership}
For every constant $\eps>0$, \eps-\GC belongs to \PPAD under the closed or open tolerance convention.
\end{lemma}

\begin{proof}
For $\eps>1/6$, membership follows from the fact that the problem is polynomial-time solvable for both conventions, i.e., due to  Corollary~\ref{cor:allbig} for the standard convention and Proposition~\ref{prop:alternative-algorithms} for the open tolerance. It therefore suffices to consider $0<\eps\le1/6$.

For $t\in\mathbb R$, write
\[
  [t]_{[0,1]}=\min\{1,\max\{0,t\}\}.
\]
For arithmetic gates, use the exact truncated output target. For a comparator,
use 

\[
  \left[
    \frac{y-x+\eps}{2\eps}
  \right]_{[0,1]}.
\]
For negation, use
\[
  \left[
    \frac{1-\eps-x}{1-2\eps}
  \right]_{[0,1]},
\]
and for $G_{\vee}$ and $G_{\wedge}$, use
\[
  \left[
    \frac{\max\{x,y\}-\eps}{1-2\eps}
  \right]_{[0,1]}
  \qquad\text{and}\qquad
  \left[
    \frac{\min\{x,y\}-\eps}{1-2\eps}
  \right]_{[0,1]},
  \text{respectively.}
\]

These functions are continuous and piecewise linear.
Every arithmetic target is satisfied exactly, while
the comparator and Boolean selectors output exactly
$0$ or $1$ whenever an output constraint is forced, under
either tolerance convention. Consequently, every fixed point of the resulting continuous piecewise-linear map satisfies all gates. Membership follows from
$\linearFIXP=\PPAD$~\cite{EY10}.
\end{proof}

We now show that the \PURE gadget continues to work correctly for the open tolerance convention not only when $0 < \eps < 1/6$ but also at $\eps = 1/6$. For simplicity, we show these analyses for the most concise of our \PURE gadget constructions, i.e., that of Appendix~\ref{app:hardness}, and note that the proofs for the construction of Section~\ref{sec:hardness} have similar arguments.

\begin{lemma}
\label{lem:convention-gadget}
Given a solution to the constructed \eps-\GC instance under open tolerance, decode a value $v < \eps$ as $0$, a value $v > 1-\eps$ as $1$, and a value $v \in [\eps, 1-\eps]$ as $\bbot$. Then, for every $0< \eps \leq 1/6$ the purification gadget implements \PURE correctly.
\end{lemma}

\begin{proof}
Under open tolerance, the constant and subtraction gates give $1-u-2\eps<r<1-u+\eps$.

If $u<\eps$, the argument of Lemma~\ref{lem:zero}
applies: all relevant comparator gaps are greater
than $\eps$, so $v,w<\eps$.

If $u>1-\eps$, the argument of Lemma~\ref{lem:one}
forces $a<\eps$ and $v>1-\eps$. Open tolerance also
gives $r<2\eps$ and $d<4\eps$, so $u-d>1-5\eps\ge\eps$.
Hence the final comparator forces $w>1-\eps$.

It remains to show that at least one output is pure
for every input. Suppose $v\le1-\eps$. The comparator
producing $v$ and the addition gate producing $s$ give
$a\ge s-\eps>\min\{2u,1\}-2\eps$.
If $u<1/2$, then $r+a >(1-u-2\eps)+(2u-2\eps) =1+u-4\eps$.
If $u\ge1/2$, then $a>1-2\eps>\eps$.
Therefore the comparator producing $a$ cannot be
forced low, and hence $r\ge u-\eps$.
Consequently,
$r+a>u+1-3\eps>1+u-4\eps$.
Moreover, combining $r\ge u-\eps$ with
$r<1-u+\eps$ gives $u<1/2+\eps$.
Thus, in both cases, $r+a>1+u-4\eps\ge u+2\eps$
and $u+2\eps<\frac12+3\eps\le1$.
It follows that $d>\min\{r+a,1\}-\eps>u+\eps$.
The final comparator therefore forces $w<\eps$.
Hence at least one output is pure, completing
the proof.
\end{proof}

\begin{proposition}
\label{prop:open}
Under open tolerance,
\eps-\GC is \PPAD-complete for every constant
$0<\eps\le1/6$.
\end{proposition}

\begin{proof}
Use the open decoding from Lemma~\ref{lem:convention-gadget}. Under this
decoding, $G_\neg$ and $G_\vee$ implement \NOT and \OR directly,
while Lemma~\ref{lem:convention-gadget} implements \PURE throughout
the stated range. Membership follows from Lemma~\ref{lem:membership}.
\end{proof}

\begin{proof}[Proof of Theorem~\ref{thm:conventions}]
For the standard convention, Theorem~\ref{thm:main} gives hardness for $\eps < 1/6$ and a polynomial-time algorithm for $\eps \geq 1/6$.

For the open tolerance convention, completeness for $\eps \leq 1/6$ follows from
Proposition~\ref{prop:open}, and
Proposition~\ref{prop:alternative-algorithms} gives polynomial-time
algorithms for $\eps > 1/6$. 
\end{proof}

\section{Algorithm analysis for any $\eps \geq 1/6$}
\label{app:eps-above-1/6}

We first run Algorithm~\ref{alg:propagation} using the state system of Section~\ref{sec:inf-states}, and
let $\sigma$ be its final state assignment. First retain the original
numerical interpretation and define
\[
\gvalonly^{\nu}(v)=\nu(\sigma(v)).
\]

\begin{lemma}\label{lem:extend7}
For every $1/6\le\eps<1/3$, the assignment
$\gvalonly^{\nu}$ is an \eps-\GC solution.
\end{lemma}

\begin{proof}
Every arithmetic error under $\gvalonly^{\nu}$ is at most $1/6$, and
every forced Boolean output has value $1/6$ or $5/6$. For
$\eps<1/3$, the Boolean-low states remain $L,\lst$, and the only
Boolean-high state remains $H$.

Finally, any comparator forced at tolerance
$\eps\ge1/6$ has an input gap greater than \eps, hence
greater than $1/6$. It was therefore forced in the same direction at
tolerance $1/6$, and its committed output satisfies the required output constraint.
\end{proof}

For larger tolerances, we keep the same information states with the same information order \eqref{eq:order} but change their numerical interpretation. Define
\begin{equation}\label{eq:mu}
\mu(L)=\mu(\lst)=\tfrac13,\qquad
\mu(M)=\mu(C)=\mu(D)=\tfrac12,\qquad
\mu(A)=\mu(H)=\tfrac23,
\end{equation}
and let $\gvalonly^{\mu}(v)=\mu(\sigma(v))$.

\begin{lemma}\label{lem:readout}
The assignment $\gvalonly^{\mu}$ is a $\tfrac13$-\GC solution.
\end{lemma}

\begin{proof}
We verify the gates without changing the state assignment.

\noindent\emph{Constants and copy.}
For every state $p$,
\[
|\mu(p)-\nu(p)|\le\tfrac16.
\]
Since a constant gate had error at most $1/6$ under $\nu$, its error under
$\mu$ is at most $1/3$. Copy gates remain exact because equal states have
equal numerical interpretations.

\noindent\emph{Addition.}
Every addition output is $A$ or $H$, both of value $2/3$ under $\mu$.
Every input has value at least $1/3$, so the truncated sum lies in
$[2/3,1]$ and is within $1/3$ of the output.

\noindent\emph{Subtraction.}
An output in $\{L,\lst\}$ has value $1/3$. Such an output occurs whenever
the first input belongs to $\mathcal I_0$ or the second belongs to
$\mathcal I_1$, or in the remaining cases not covered by the rules for
$C$ and $D$; direct inspection gives a subtraction target at most
$1/2$, and hence error at most $1/3$. The outputs $C,D$ have value
$1/2$ and occur only when the first input is $H$ and the second is $M$
or belongs to $\mathcal I_0$. Their subtraction targets are $1/6$ or
$1/3$, so the error is again at most $1/3$.

\noindent\emph{Multiplication.}
If $c\le2/5$, the output has value $1/3$ and the target lies in
$[0,4/15]$. If $c\ge4/5$, copying the input gives error at most
$\tfrac15\cdot\tfrac23=\tfrac2{15}$.
In the middle interval, the possible inputs under $\mu$ are
$1/3,1/2,2/3$, while the corresponding outputs under $P$ are
$1/3$ or $1/2$. Directly over $2/5<c<4/5$, the maximum error is
$3/10<1/3$.

\noindent\emph{Boolean gates.}
At tolerance $1/3$, the numerical low states are $L,\lst$, and the
numerical high states are $A,H$. Every numerically low state belongs to \(\mathcal I_0\), and every numerically high state belongs to \(\mathcal I_1\). Their output values
are $1/3$ and $2/3$, so all forced output constraints are satisfied.

\noindent\emph{Comparators.}
At tolerance $1/3$, a comparator is never forced, since no two
values under $\mu$ differ by more than $1/3$.
\end{proof}

\begin{corollary}\label{cor:allbig}
For every $\eps\ge1/6$, \eps-\GC can be solved in
polynomial time.
\end{corollary}

\begin{proof}
For $1/6\le\eps<1/3$, use $\gvalonly^{\nu}$ and
Lemma~\ref{lem:extend7}. For $1/3\le\eps<1/2$, use
$\gvalonly^{\mu}$. Its arithmetic errors are at most $1/3$. The only
values below $1/2$ are $1/3$, and the only values above $1/2$ are
$2/3$, so every Boolean premise that becomes active in this range is
consistent with the information class used by the corresponding Boolean
rule. Moreover, increasing the comparator tolerance cannot force a new
comparator pair.

For $\eps\ge1/2$, the Boolean gates essentially break down. By assigning every variable the value $1/2$ all the gates are satisfied: arithmetic errors are at most $1/2$, comparator inputs are equal, and the value $1/2$ satisfies every forced Boolean output.
\end{proof}

Theorem~\ref{thm:main} follows from Theorem~\ref{thm:hardness} and
Corollary~\ref{cor:allbig}.

\section{Proofs for violated gates}\label{sec:robust}
Algorithm~\ref{alg:overview} summarizes both violated-gate algorithms.
The graph argument of Section~\ref{subsec:frag-overview} is proved in
Appendices~\ref{subsec:quarter} (quarter deletion) and~\ref{subsec:fragment}
(decomposition), and Appendix~\ref{subsec:lower} gives the lower bound.
The remaining arguments are circuit-specific:
Appendices~\ref{subsec:components} and~\ref{subsec:balancing} treat \GC,
and Appendix~\ref{subsec:pure-circuit} treats fan-out-one \PC.
We use the standard definition of \GC throughout.
By Corollary~\ref{cor:allbig}, every gate can be
satisfied when $\eps\ge1/6$. Thus, it suffices
to consider $0<\eps<1/6$.

\begin{algorithm}[t]
\caption{The violated-gate algorithms of Theorems~\ref{thm:robust}
and~\ref{thm:pc-robust}.}\label{alg:overview}
\begin{algorithmic}[1]
\Function{ViolatedGCircuit}{circuit, constants $\eps,\tau>0$}
    \If{$\eps\ge1/6$}
        \State \Return the solution of Corollary~\ref{cor:allbig}
    \EndIf
    \State build the dependency digraph $D$;
         $\lambda\gets\lceil4/\tau\rceil$
  \State $(S_0,\pi)\gets\Call{QuarterDeletion}{D}$
         \Comment{Algorithm~\ref{alg:quarter}}
  \State $X\gets$ decomposition set for $D-S_0$, $\pi$, and $\lambda$
         \Comment{Lemma~\ref{lem:fragment}}
  \State give constant variables their prescribed values, and the
         variables in $S_0\cup X$ the value $0$
  \For{each strongly connected component $C$ of $D-(S_0\cup X)$, in
       topological order}
    \State assign $C$ by exhaustive search over the fixed grid
           \Comment{Lemma~\ref{lem:small-component}}
  \EndFor
  \State let $\gvalonly^{\mathrm{decomp}}$ be the resulting assignment
  \State construct $\gvalonly^{\mathrm{const}}$ and
         $\gvalonly^{\mathrm{zero}}$ as in
         Lemma~\ref{lem:dep-cert}
  \State \Return whichever of $\gvalonly^{\mathrm{decomp}}$,
         $\gvalonly^{\mathrm{const}}$, and
         $\gvalonly^{\mathrm{zero}}$ violates the fewest gates
\EndFunction
\Statex
\Function{FanOutOnePureCircuit}{instance, integer $\lambda\ge1$}
  \State build the gate digraph $D$
  \State $(S_0,\pi)\gets\Call{QuarterDeletion}{D}$
  \State $X\gets$ decomposition set for $D-S_0$, $\pi$, and $\lambda$
         \Comment{Lemma~\ref{lem:fragment}}
  \State give the outputs of the gates in $S_0\cup X$ arbitrary values
  \For{each strongly connected component $C$ of $D-(S_0\cup X)$, in
       topological order}
    \State assign the nodes produced by $C$ by exhaustive search over
           $\{0,1,\bbot\}$
           \Comment{Lemma~\ref{lem:pc-boundary-totality}}
  \EndFor
  \State \Return the resulting assignment
\EndFunction
\end{algorithmic}
\end{algorithm}

\subsection{Quarter deletion and bounded-width elimination orderings}
\label{subsec:quarter}

We first show that deleting at most one quarter of the vertices
of in-degree two suffices to obtain an elimination ordering
of bounded in-width.

   We use directed vertex elimination~\cite{RoseTarjan78}, with the
   width notion of Hunter and Kreutzer~\cite{HK08} applied to the
   reverse digraph.
   
\begin{definition}[Directed elimination]\label{def:elimination}
Eliminating a vertex $v$ removes it and adds an arc $uw$
for every current in-neighbor $u$ and out-neighbor $w$
of $v$, suppressing loops and parallel arcs.
An \emph{elimination ordering} is an ordering of the vertices,
whose \emph{in-width} is the maximum in-degree of a vertex
immediately before its elimination.
An elimination is \emph{free} if the vertex has in-degree
at most one.
\end{definition}

Unlike elimination, \emph{deletion} removes a vertex without
adding shortcuts. Deleted vertices do not appear in the final
elimination ordering.

We collect the elimination properties needed below.
Item~(3) is one direction of the fill-path characterization of Rose and Tarjan~\cite{RoseTarjan78}. Item~(4) follows from the support characterization of
Hunter and Kreutzer~\cite{HK08}, applied to
the reverse digraph.
For completeness, we include standalone proofs even if some of the results needed come from prior work.

\begin{lemma}\label{lem:replay}
\leavevmode
\begin{enumerate}
\item Eliminating a vertex of in-degree at most one does not
increase any remaining in-degree. Neither deletion nor
elimination removes an arc between remaining vertices.

\item Moving a deletion earlier, while preserving the order
of all other operations, cannot increase the in-degree at
any subsequent elimination.

\item If a path from $u$ to $w\ne u$ has all its internal
vertices eliminated before either endpoint is removed,
then the arc $uw$ is present after these eliminations.

\item Let $v_1,\ldots,v_n$ be an elimination ordering of
in-width at most $k$, and let $P_i=\{v_1,\ldots,v_i\}$.
For every $i$, at most $k$ vertices $z\notin P_i$
can reach $v_i$ along a path whose internal vertices
belong to $P_{i-1}$. This property is preserved when
restricting the ordering to an induced subgraph.
\end{enumerate}
\end{lemma}

\begin{proof}
For (1), eliminating a vertex with a unique in-neighbor
replaces each outgoing arc by at most one shortcut.
Thus, no remaining in-degree increases.

For (2), compare the two executions on their common
remaining vertices. Moving a deletion earlier only
removes arcs, and every shortcut requires two existing
arcs. Inductively, the execution with the earlier deletion
has a subset of the arcs of the original execution.
Hence its in-degrees cannot be larger.

For (3), eliminate the internal vertices of the path
one by one. Each elimination replaces the two incident
path arcs by a shortcut, eventually producing $uw$.
This arc remains until an endpoint is removed.

For (4), eliminate $v_1,\ldots,v_{i-1}$.
By (3), every stated path produces an arc $zv_i$.
Since $v_i$ then has in-degree at most $k$, there
are at most $k$ such vertices $z$.
Restricting the ordering to an induced subgraph
cannot introduce new paths of this form.
\end{proof}

Our ordinary deletion follows the charging of~\cite{BorowieckiDS14,Meyer26}
(there, on acyclic orientations); seeds, marks and repairs handle
directed cycles. 

Our algorithm maintains a working digraph $R$, a deletion
set $S$, an ordering $\pi$, and a set $Q$ of at most two
marked vertices. It repeatedly eliminates vertices of
in-degree at most one. When none remain, it performs
one of three operations:

\begin{itemize}
\item An \emph{ordinary deletion} removes a vertex of
out-degree at least three.

\item A \emph{repair} deletes a marked vertex if no
ordinary deletion is possible.

\item A \emph{seed} removes an arbitrary vertex without
adding shortcuts, appends it to $\pi$, and marks its
in-neighbors. A seed is introduced only when $Q$ is empty.
\end{itemize}

When a marked vertex is freely eliminated, its mark is
transferred to its unique current in-neighbor, if one exists.
The marks will account for the shortcuts omitted when
seeds are removed.

\begin{algorithm}[t]
\caption{Quarter deletion.}\label{alg:quarter}
\begin{algorithmic}[1]
\Require simple digraph $D$ of maximum in-degree at most two
\Ensure deletion set $S$ and elimination ordering $\pi$ of $D-S$
\Function{QuarterDeletion}{$D$}
  \State $R\gets D$; $S,Q\gets\varnothing$; $\pi\gets$ empty ordering
  \Loop
    \While{some $v\in V(R)$ has $d_R^-(v)\le1$}
      \State $I\gets N_R^-(v)$
      \State eliminate $v$ in $R$; append $v$ to $\pi$
             \label{line:free-elim}
      \If{$v\in Q$}
        \State $Q\gets(Q\setminus\{v\})\cup I$
               \label{line:transfer}
      \EndIf
    \EndWhile
    \If{$R$ is empty}
      \State \Return $(S,\pi)$
    \EndIf
    \If{some $x$ has $d_R^+(x)\ge3$}
      \State choose such an $x$
      \State delete $x$ from $R$; $S\gets S\cup\{x\}$
      \State $Q\gets Q\setminus\{x\}$
             \label{line:ordinary}
    \ElsIf{$Q\ne\varnothing$}
      \State choose $x\in Q$
      \State delete $x$ from $R$; $S\gets S\cup\{x\}$
      \State $Q\gets Q\setminus\{x\}$
             \label{line:repair}
    \Else
      \State choose any $s\in V(R)$
      \State $Q\gets N_R^-(s)$
             \label{line:seed-mark}
      \State remove $s$ from $R$ without shortcuts;
             append $s$ to $\pi$
             \label{line:seed-elim}
    \EndIf
  \EndLoop
\EndFunction
\end{algorithmic}
\end{algorithm}

The working digraph always has maximum in-degree at most two:
free eliminations do not increase in-degrees by
Lemma~\ref{lem:replay}(1), and the other operations add
no arcs.

\begin{theorem}\label{thm:quarter-del}
Let $D$ be a simple digraph of maximum in-degree at most
two, with $n_2$ vertices of in-degree two. In polynomial
time one can find a set $S\subseteq V(D)$ with
$|S|\le n_2/4$ and an elimination ordering of $D-S$
of in-width at most three.
\end{theorem}

\begin{proof}
We prove the in-width and deletion bounds separately.

\smallskip
\noindent\emph{The in-width bound.}
The working graph $R$ omits shortcuts when seeds are removed.
To account for them, consider a second digraph $\widehat R$
that performs the genuine eliminations of all vertices
appended to $\pi$, while making the same deletions as $R$.

We maintain the invariant
\begin{equation}\label{eq:marked-invariant}
\begin{gathered}
V(\widehat R)=V(R),\qquad E(R)\subseteq E(\widehat R),\\
E(\widehat R)\setminus E(R)\subseteq Q\times V(R),
\qquad Q\subseteq V(R),\quad |Q|\le2.
\end{gathered}
\end{equation}
Thus, every omitted arc has a marked tail.

Initially $\widehat R=R=D$ and $Q=\varnothing$, so the
invariant holds.

Suppose first that $v$ is freely eliminated in $R$.
Every additional in-neighbor of $v$ in $\widehat R$
belongs to $Q$. Hence
\[
d_{\widehat R}^-(v)
\le d_R^-(v)+|Q|
\le3.
\]

To preserve the invariant, consider a shortcut $uw$
created in $\widehat R$ but absent from $R$ after
eliminating $v$. If $uv\notin E(R)$, then $u\in Q$.
Otherwise, $uv\in E(R)$ and $vw\notin E(R)$,
so $v\in Q$. Since $v$ has at most one in-neighbor
in $R$, transferring its mark to that in-neighbor
marks $u$. Thus, every omitted shortcut has a marked
tail. Previously omitted arcs whose endpoints survive
retain their marked tails, and transferring a mark
does not increase $|Q|$.
The invariant is therefore preserved.

An ordinary deletion or repair removes the same vertex
from both graphs and creates no shortcuts, so it also
preserves the invariant.

Finally, a seed is chosen only when $Q=\varnothing$,
and hence $\widehat R=R$. At this point every remaining
vertex has in-degree two. Thus, the seed is eliminated
in $\widehat R$ at in-degree two. Its two in-neighbors
are marked in $R$, accounting for every shortcut omitted
there and preserving the invariant.

Consequently, every vertex appended to $\pi$ is
eliminated in $\widehat R$ at in-degree at most three.
Moving all deletions to the beginning cannot increase
any subsequent in-degree by Lemma~\ref{lem:replay}(2).
Hence $\pi$ is an elimination ordering of $D-S$
of in-width at most three.

\smallskip
\noindent\emph{The deletion bound.}
The initial free eliminations remove every vertex
whose original in-degree is at most one. Let $m$
be the number of remaining vertices. Then
\[
m\le n_2.
\]

Call each ordinary deletion, repair, or seed, together
with the free eliminations that follow it, a \emph{round}.
At the beginning of each round, every remaining vertex
has in-degree two.

An ordinary deletion removes a vertex of out-degree at
least three. Its deletion makes at least three
out-neighbors eligible for free elimination, so the
round removes at least four vertices.

When no ordinary deletion is possible, every out-degree
is at most two. Since every in-degree is exactly two,
the equality of the sums of in-degrees and out-degrees
implies that every out-degree is also two.
A repair or seed therefore makes two out-neighbors
eligible for free elimination, so its round removes
at least three vertices.

Let $a,h,t$ be the numbers of ordinary deletions,
repairs, and seeds, respectively. Each seed creates
at most two marks, transferring marks does not
increase their number, and each repair consumes one.
Hence $h\le2t$.
Since the rounds remove disjoint sets of vertices, and $h\le 2t$ gives $3t\ge h$,
\[
m\ge4a+3h+3t\ge4(a+h).
\]
Only ordinary deletions and repairs contribute to $S$.
Therefore,
\[
|S|=a+h\le\frac m4\le\frac{n_2}{4}.
\]

Every operation removes a vertex from $R$, and all
graph operations can be performed in polynomial time.
\end{proof}

\begin{remark}
Making the digraph acyclic can require deleting
$2n/3$ vertices, as shown by disjoint bidirected
triangles. Our weaker goal of in-width at most $3$
requires at most $n_2/4$ deletions and suffices
for the decomposition in the next subsection.
\end{remark}

\subsection{Decomposition from an elimination ordering}
\label{subsec:fragment}

We now use the reachability property of
Lemma~\ref{lem:replay}(4) to turn an elimination ordering
of bounded in-width into a decomposition with small
strongly connected components. The argument is analogous
to the bounded-treewidth fragmentation lemma of
Gupta et al.~\cite{GLLMW19}.

\begin{lemma}\label{lem:fragment}
Given an $n$-vertex simple digraph $D$, an elimination ordering of
$D$ of in-width at most $k$, and an integer $\lambda\ge1$,
one can compute in polynomial time a set $X\subseteq V(D)$ with
\[
  |X|\le\frac{(k+1)n}{\lambda+1}
\]
such that every strongly connected component of $D-X$
has at most $\lambda$ vertices.
\end{lemma}

\begin{proof}
We maintain an induced subgraph $H$ of $D$, initially $D$,
and a deletion set $X$, initially empty. The ordering on $H$
is always the restriction of the original ordering, so
Lemma~\ref{lem:replay}(4) continues to apply.

If every strongly connected component of $H$ has at most
$\lambda$ vertices, we stop. Otherwise, let $v_i$ be the
first vertex in the restricted ordering whose prefix $P_i$
contains a strongly connected component $C$ of more than
$\lambda$ vertices. Necessarily $v_i\in C$. Let
\[
  W=\{u\in P_i:u\leadsto v_i\text{ in }H[P_i]\},
\]
where reachability includes paths of length zero.
Since $C\subseteq W$, we have $|W|\ge\lambda+1$.

Let $B$ be the vertices outside $W$ with an arc into $W$.
Every $z\in B$ occurs after $v_i$ in the ordering, since
otherwise $z$ would belong to $W$. Moreover, $z$ can reach
$v_i$ along a path whose internal vertices precede $v_i$.
By Lemma~\ref{lem:replay}(4), $|B|\le k$.

Add $B\cup\{v_i\}$ to $X$, set aside
$W\setminus\{v_i\}$, and continue on $H-(W\cup B)$.
By the choice of $v_i$, every strongly connected
component of $H[W\setminus\{v_i\}]$ has at most
$\lambda$ vertices. Furthermore, deleting $B$ removes
every arc entering $W$ from the remaining graph.
Thus, no strongly connected component of $D-X$ can
intersect both a set-aside part and vertices processed
later. Consequently, every strongly connected component
of $D-X$ has at most $\lambda$ vertices.

Each iteration removes at least $\lambda+1$ vertices
from $H$ and adds at most $k+1$ vertices to $X$.
Hence there are at most $n/(\lambda+1)$ iterations, giving
\[
  |X|\le\frac{(k+1)n}{\lambda+1}.
\]
All operations can be performed in polynomial time.
\end{proof}

\begin{proof}[Proof of Theorem~\ref{thm:fragmentation}]
By Theorem~\ref{thm:quarter-del}, we can compute a set
$S_0$ with $|S_0|\le n_2/4$ such that $D-S_0$
admits an elimination ordering of in-width at most three.
Applying Lemma~\ref{lem:fragment} to $D-S_0$ with
$k=3$ gives a set $X$ with
$|X|\le4n/(\lambda+1)$ such that every strongly
connected component of $D-(S_0\cup X)$ has at most
$\lambda$ vertices. Taking $S=S_0\cup X$ completes
the proof.
\end{proof}

\subsection{A linear lower bound}\label{subsec:lower}
The next lemma shows that a linear number of deletions can be
necessary. Its proof uses the first-moment method~\cite{AlonSpencer}
to rule out nearly balanced cuts with no arcs in one direction.

\begin{lemma}
\label{lem:fragment-lb}
With natural logarithms, let $\gamma\approx 0.0819$ be the unique
root in $(0,1/10)$ of
\[
  \Phi(x)
  =-(1-x)\log(1-x)+2(1+x)\log(1+x)
   -(1+3x)\log2-3x\log x.
\]
For every nonnegative constant $\alpha<\gamma$ and every function
$\lambda(n)=o(n)$, there are arbitrarily large simple $n$-vertex
digraphs $D$ of maximum in-degree at most two such that every set
$S\subseteq V(D)$ for which all strongly connected components of $D-S$
have at most $\lambda(n)$ vertices satisfies $|S|>\alpha n$.
\end{lemma}

\begin{proof}
Let $\phi_1,\phi_2$ be independent uniformly random permutations of
$[n]$, and let $D$ have the arcs $v\to\phi_j(v)$ for $j\in\{1,2\}$, with
loops and parallel arcs suppressed; this does not affect reachability.
Then $D$ is simple and has maximum in-degree at most two.

Suppose that some set of at most $\alpha n$ vertices leaves only strongly
connected components with at most $\lambda(n)$ vertices. Enlarging the set
preserves this property, so we may take $S$ with $|S|=s=\lfloor\alpha
n\rfloor$. List the strongly connected components of $D-S$ in topological
order, let $A$ be the union of the shortest prefix with at least
$(n-s)/2$ vertices, and let $B=[n]\setminus(A\cup S)$. Then no arc goes
from $B$ to $A$, and $|A|,|B|=(n-s)/2+O(\lambda(n))$.

Fix part sizes $a,b,s$. For a uniformly random permutation $\phi$, the set
$\phi(B)$ is a uniformly random $b$-subset of $[n]$, so $\phi$ creates no
arc from $B$ to $A$ with probability $\binom{n-a}{b}/\binom{n}{b}$. By
independence, the expected number of such partitions $(A,B,S)$ is
\[
  \binom{n}{a,b,s}
  \left(\frac{\binom{n-a}{b}}{\binom{n}{b}}\right)^2
  =
  \frac{(a+s)!^2(b+s)!^2}{n!\,a!\,b!\,(s!)^3}
  =\exp\bigl(n(\Phi(\alpha)+o(1))\bigr),
\]
by Stirling's formula, with $0\log0=0$; since $a/n,b/n\to(1-\alpha)/2$ and
$s/n\to\alpha$, the error term is uniform over the permitted sizes.

Since $\Phi(0)=-\log2$, $\Phi(1/10)>0$, and
$\Phi'(x)=\log\frac{(1-x)(1+x)^2}{8x^3}>0$ on $(0,1/10)$, the constant
$\gamma$ is well defined and $\Phi(\alpha)<0$. Summing over the at most
$n+1$ values of $a$, the expected number of such partitions is at most
$(n+1)\exp\bigl(n(\Phi(\alpha)+o(1))\bigr)=o(1)$. By Markov's inequality,
for all sufficiently large $n$ some outcome of $\phi_1,\phi_2$ yields a
digraph with no such partition, and hence with the required property.
\end{proof}

\subsection{Solving the small components}\label{subsec:components}

We return to the dependency digraph of Definition~\ref{def:dependency}.
The decomposition makes its strongly connected components small, but it
does not make the graph acyclic. We therefore cannot, in general, assign
the variables of a component one gate at a time. Instead, once the inputs
entering a component have been fixed, we solve all its variables
simultaneously. Suppressing loops and parallel arcs affects only the
graph: all original gate constraints, including those of gates with
self-inputs, are retained.

For a set $C$ of variables, a \emph{boundary input} is an input outside
$C$ to a gate producing a variable in $C$. The next lemma shows that, for
any fixed boundary inputs, a satisfying assignment exists on a fixed
finite grid. Brouwer's theorem is used only to prove that it exists; the
algorithm itself enumerates the grid.

\begin{lemma}\label{lem:small-component}
Let $0<\eps\le1/6$, let $g=\lceil\eps^{-2}\rceil$, and let
\[
  \Gamma=\{0,1/g,2/g,\ldots,1\}.
\]
Let $C$ be a set of $k$ non-constant variables, and fix rational values in
$[0,1]$ for all boundary inputs of the gates producing variables of $C$.
Under the standard definition of \GC, some assignment in
$\Gamma^C$ satisfies all these gates. One can be found by testing at most
$(g+1)^k$ assignments, each test taking time polynomial in the encoding
length of the gates and the boundary values.
\end{lemma}

\begin{proof}
Fix the boundary inputs. The continuous gate functions used in the proof
of Lemma~\ref{lem:membership} define a map
$\Psi:[0,1]^C\to[0,1]^C$. Since the arithmetic, comparator, and
Boolean gate functions have Lipschitz constants at most $2$, $1/\eps$,
and $1/(1-2\eps)$, respectively, $\Psi$ is $1/\eps$-Lipschitz in
the sup norm.

By Brouwer's theorem, $\Psi$ has a fixed point $\gvalonly^*$.
Round each coordinate of $\gvalonly^*$ to its nearest point in
$\Gamma$, obtaining $\widehat{\gvalonly}$. Since
$g=\lceil\eps^{-2}\rceil$,
\[
\begin{aligned}
\|\widehat{\gvalonly}
  -\Psi(\widehat{\gvalonly})\|_\infty
&\le \frac{1+1/\eps}{2g}\\
&\le \frac{\eps(1+\eps)}2
<\eps.
\end{aligned}
\]

The gate functions return the exact target for arithmetic gates and
the required Boolean value whenever a comparator or Boolean premise
is triggered. Hence the residual bound implies that every gate is
$\eps$-satisfied.

Thus, a satisfying assignment exists in $\Gamma^C$ and can be found
by enumerating at most $(g+1)^k$ assignments.
\end{proof}

\begin{lemma}\label{lem:struct-cert}
For every fixed $0<\eps\le1/6$ and $\tau>0$, an assignment
$\gvalonly^{\mathrm{decomp}}$ with
\[
  \mathrm{viol}(\gvalonly^{\mathrm{decomp}})
  \le\frac{N_2}{4}+\tau N
\]
can be computed in polynomial time.
\end{lemma}

\begin{proof}
We apply Theorem~\ref{thm:fragmentation} to the dependency digraph $D$
with $\lambda=\lceil4/\tau\rceil$. Since $D$ has $N-N_0$ vertices, of
which $N_2$ have in-degree two, this gives a set $S$ with
\[
  |S|\le\frac{N_2}{4}+\frac{4(N-N_0)}{\lambda+1}
     \le\frac{N_2}{4}+\tau N
\]
such that every strongly connected component of $D-S$ has at most
$\lambda$ vertices.

We give the constant variables their prescribed values and the
variables in $S$ the value zero. We then process the components of
$D-S$ in topological order. Every boundary input is a constant
variable, a variable in $S$, or a variable in an earlier component,
so its value has already been fixed.

By Lemma~\ref{lem:small-component}, each component can be assigned
by testing at most $(g+1)^\lambda$ assignments. Since previously
assigned values are never changed, only gates producing variables
in $S$ may be violated. For fixed $\eps$ and $\tau$, the algorithm
runs in polynomial time, and all assigned values have polynomial
encoding length.
\end{proof}

The term $\tau N$ is measured against all gates, since the deletions
made by Lemma~\ref{lem:fragment} may also include outputs of gates
with at most one dependency.

\subsection{Complementary assignments and the final bound}
\label{subsec:balancing}

Section~\ref{subsec:gc-overview} obtains the bound $1/5$ by
comparing $\gvalonly^{\mathrm{decomp}}$ with
$\gvalonly^{\mathrm{const}}$. We now use a third assignment,
$\gvalonly^{\mathrm{zero}}$, which sets every variable to zero.
Comparing these three assignments gives a bound that also
depends on the number of negation gates.

Let $N_\neg$ denote the number of negation gates $G_\neg$ in
the circuit. Every negation gate has one input and hence at
most one dependency, so $N_\neg\le N_1$.

\begin{lemma}\label{lem:dep-cert}
Let $\gvalonly^{\mathrm{const}}$ give every constant variable its
prescribed value and every other variable the value zero, and let
$\gvalonly^{\mathrm{zero}}$ give every variable the value zero. Then
\[
  \mathrm{viol}(\gvalonly^{\mathrm{const}})\le N_1
  \qquad\text{and}\qquad
  \mathrm{viol}(\gvalonly^{\mathrm{zero}})\le N_0+N_\neg.
\]
\end{lemma}

\begin{proof}
The assignment $\gvalonly^{\mathrm{const}}$ satisfies every constant
gate. A gate whose output has two dependencies uses both input positions
for distinct non-constant variables, so it is a $G_+$, $G_-$, $G_<$,
$G_\vee$ or $G_\wedge$ gate, and both its inputs and its output are zero.
Addition and subtraction then have target zero, the comparator is
unforced, and zero satisfies the forced low conclusion of $G_\vee$ and
$G_\wedge$.
Thus, only the $N_1$ other non-constant gates may fail.

Under $\gvalonly^{\mathrm{zero}}$, every non-constant arithmetic gate has
target zero, every comparator is unforced, and zero satisfies every forced
conclusion of $G_\vee$ and $G_\wedge$. Thus, only constant and negation
gates may fail.
\end{proof}

\begin{lemma}
\label{lem:robust-rho}
Let $\eps>0$ and $\rho\in[0,1]$ be constants, and suppose that at
most a $\rho$-fraction of the gates are negation gates. $(\eps,\delta)$-\GC can be solved in
polynomial time for every constant
\[
  \delta>\min\left\{\frac15,\frac{1+\rho}{6}\right\}.
\]
\end{lemma}

\begin{proof}
For $\eps\ge1/6$, use Corollary~\ref{cor:allbig}.
It therefore suffices to consider $0<\eps<1/6$.
The empty circuit is trivial.

\smallskip
\noindent\emph{The violation bound.}
Set
\[
  \theta_\rho=\min\{1/5,(1+\rho)/6\},\qquad
  \Delta=\delta-\theta_\rho>0,\qquad
  \tau=\Delta/2.
\]
Let $\gvalonly^\star$ be the assignment with the fewest violations among
the three assignments of Lemmas~\ref{lem:struct-cert}
and~\ref{lem:dep-cert}, and write $v=\mathrm{viol}(\gvalonly^\star)$.
Then
\[
  v\le\frac{N_2}{4}+\tau N,\qquad
  v\le N_1,\qquad
  v\le N_0+N_\neg.
\]
Four times the first bound plus the second gives
\[
  5v\le N_2+N_1+4\tau N\le(1+4\tau)N.
\]
Adding the third bound as well, and using $N_\neg\le\rho N$, gives
\[
  6v\le N_2+N_1+N_0+N_\neg+4\tau N\le(1+\rho+4\tau)N.
\]
Consequently,
\[
  \frac vN
  \le\min\left\{
    \frac15+\frac{4\tau}{5},\
    \frac{1+\rho}{6}+\frac{2\tau}{3}
  \right\}
  \le\theta_\rho+\tau
  =\theta_\rho+\frac\Delta2
  <\delta.
\]

\smallskip
\noindent\emph{The running time.}
Let $\ell$ be the binary encoding length of the circuit, including its
rational parameters, and recall that
\[
  g=\lceil\eps^{-2}\rceil,
  \qquad
  \lambda=\lceil4/\tau\rceil=\lceil8/\Delta\rceil.
\]
There are at most $N$ components, each requiring at most $(g+1)^\lambda$
candidate assignments, and each candidate can be checked in time
polynomial in $\ell$. The graph operations take polynomial time with an
exponent independent of $\lambda$. The total running time is therefore
\[
  (g+1)^{O(1/\Delta)}\operatorname{poly}(\ell).
\]
\end{proof}

\begin{proof}[Proof of Theorem~\ref{thm:robust}]
Every circuit satisfies the hypothesis of Lemma~\ref{lem:robust-rho} with $\rho=1$, in which case $\min\{1/5,(1+\rho)/6\}=1/5$.
\end{proof}

These three estimates alone cannot give a general threshold
below $1/5$. Indeed, for $N_0=0$, $N_1=N_\neg=N/5$, and
$N_2=4N/5$, their right-hand sides are $N/5+\tau N$,
$N/5$, and $N/5$, respectively. Their minimum is therefore
$N/5$ for every $\tau>0$. This calculation concerns what
the estimates certify; the assignments themselves may
violate fewer gates.



\subsection{Fan-out-one \texorpdfstring{\PC}{Pure-Circuit}}
\label{subsec:pure-circuit}

The argument for fan-out-one \PC follows the same pattern as for \GC:
we decompose the gate digraph and solve the remaining components one at a
time. Two differences make it simpler. The gate counts are determined by
the fan-out-one condition, so no balancing is needed, and a component can
be solved over the three values $\{0,1,\bbot\}$ rather than a grid.

Consider a fan-out-one instance of Definition~\ref{def:delta-pc}, and let
$m_P$, $m_N$, and $m_R$ be the numbers of \PURE gates, \NOT gates, and
two-input gates (\OR, \NOR, \AND, \NAND), respectively, with
$m=m_P+m_N+m_R$. Every node is produced once and occupies exactly one
input position. Counting outputs and input positions therefore gives
\begin{equation}\label{eq:pc-balance}
  2m_P+m_N+m_R=m_P+m_N+2m_R,
  \qquad\text{and hence}\qquad
  m_R=m_P\le\frac m2.
\end{equation}
In the gate digraph, arcs point from producers to consumers, so only
two-input gates can have in-degree two, and $n_2\le m_R\le m/2$.
Theorem~\ref{thm:fragmentation} then gives a set $S$ of gates with
\begin{equation}\label{eq:pc-fragment}
  |S|\le\frac{n_2}{4}+\frac{4m}{\lambda+1}
     \le\left(\frac18+\frac4{\lambda+1}\right)m
\end{equation}
such that every remaining strongly connected component has at most
$\lambda$ gates.

It remains to show that a component can be solved once the values
entering it are fixed. The following lemma does not use the fan-out-one
condition, which is needed only for the gate-count identity above.

\begin{lemma}\label{lem:pc-boundary-totality}
Let $C$ be a set of gates of an instance of
Definition~\ref{def:delta-pc}, and fix arbitrary values in
$\{0,1,\bbot\}$ for all inputs of gates in $C$ that are produced outside
$C$. Then the nodes produced by $C$ can be assigned values in
$\{0,1,\bbot\}$ so that every gate of $C$ is satisfied.
\end{lemma}

\begin{proof}
We represent the fixed values $0,1,\bbot$ by $0,1,1/2$, respectively, and
use the continuous functions
\[
  \begin{aligned}
  \PURE:\quad &f_1(x)=\max\{0,2x-1\}, & f_2(x)&=\min\{1,2x\},\\
  \NOT:\quad &1-x, & \NAND:\quad &\min\{1,2-x-y\},\\
  \OR:\quad &\max\{x,y\}, & \NOR:\quad &1-\max\{x,y\},\\
  \AND:\quad &\min\{x,y\}.
  \end{aligned}
\]
After substituting the fixed boundary values, these functions
define a continuous map from $[0,1]^k$ to itself, where $k$ is
the number of nodes produced by $C$. Each coordinate is defined
exactly once, since every node has a unique producer and the
two outputs of each \PURE gate are distinct. By
Brouwer's theorem it has a fixed point. We decode $0$ and $1$ as the
corresponding pure bits, and every value in $(0,1)$ as $\bbot$.

Each Boolean function returns exactly $0$ or $1$ whenever its pure
inputs force that output in Figure~\ref{fig:pc}: an input $1$ forces
$1$ for \OR and $0$ for \NOR, an input $0$ forces $0$ for \AND and $1$
for \NAND, and in the remaining forced cases all inputs are pure and the
function returns the Boolean value. When the inputs force nothing, any
output is allowed. For \PURE, both outputs copy the inputs $0$ and $1$.
For an input $x\in(0,1)$, we have $f_1(x)=0$ if $x\le1/2$ and
$f_2(x)=1$ if $x\ge1/2$, so at least one output is pure. Hence the
decoded assignment satisfies every gate of $C$ and agrees with the fixed
boundary values.
\end{proof}

\begin{proof}[Proof of Theorem~\ref{thm:pc-robust}]
The empty instance is trivial. Fix $\delta>1/8$ and set
\[
  \lambda=\left\lceil\frac{4}{\delta-1/8}\right\rceil.
\]
Then $4/(\lambda+1)<\delta-1/8$. By
\eqref{eq:pc-fragment}, we can compute a set $S$ of gates
with $|S|<\delta m$ such that every strongly connected
component of the remaining gate digraph has at most
$\lambda$ gates.

Assign arbitrary values to the outputs of the gates in
$S$, and process the remaining components in topological
order. When a component is reached, all its boundary inputs
have been fixed. A component of at most $\lambda$ gates
produces at most $2\lambda$ nodes, so
Lemma~\ref{lem:pc-boundary-totality} guarantees that
exhaustive search over at most $3^{2\lambda}$ assignments
finds one satisfying all its gates.

The assignments to earlier components are never changed.
Thus, only the gates in $S$ may be violated. There are at
most $m$ components, and the running time is
\[
  3^{2\lambda}\operatorname{poly}(m)
  =2^{O(1/(\delta-1/8))}\operatorname{poly}(m),
\]
which is polynomial for every fixed $\delta>1/8$.
\end{proof}

\section{Applications of our positive results}\label{app:positive-applications}
We give the definitions and proofs for the applications in Section~\ref{sec:positive-applications},
following the same order. The error mapping uses the numerical
$1/6$ guarantee of Theorem~\ref{thm:main}; the well-supported
public-goods result uses any fixed tolerance below $1/2$. 
The remaining applications run
Algorithm~\ref{alg:propagation} on its three-state
Boolean subsystem, sometimes with application-specific
decoding, or use Theorem~\ref{thm:fragmentation}.

\subparagraph{Approximate Nash equilibria.}
For a finite game, we write $A_i$ for player $i$'s action set, $s_i$ for
her (mixed) strategy, i.e., a probability distribution of $A_i$, and $u_i(a,s_{-i})$ for the expected payoff from
pure action $a$. Mixed strategies of different players are independent.
We use the following standard notions~\cite{PureCircuit}.

\begin{definition}[Approximate Nash equilibrium]\label{def:app-ne}
A strategy profile $s$ is an \emph{$\eps$-Nash equilibrium}
($\varepsilon$-NE) if, for every player $i$,
\[
 \operatorname{Reg}_i(s):=
 \max_{a\in A_i}u_i(a,s_{-i})-u_i(s)\le\eps.
\]
It is an \emph{$\eps$-well-supported Nash equilibrium}
($\eps$-WSNE) if, for every player $i$ and every
$a\in\operatorname{supp}(s_i) := \{j \in A_i \; : \; s_i(j) > 0\}$,
\[
 u_i(a,s_{-i})\ge\max_{b\in A_i}u_i(b,s_{-i})-\eps.
\]
The latter inequalities imply the former by averaging over $s_i$.
\end{definition}

\subparagraph{Boolean state propagation.}
For circuits containing only $G_0,G_1,G_=,G_\neg,G_\vee,G_\wedge$,
choose $R(0)=\lst$ and $R(1)=H$ in the constant rule of
Section~\ref{sec:localrules}. Starting from $M$, the states then remain
in $\{\lst,M,H\}$, where $\lst$ encodes zero, $H$ encodes one,
and $M$ is unresolved. We call $\lst$ and $H$ \emph{terminal}, as they
are maximal in $\preceq$. Let us denote by $N_i$ the in-degree of node $i$. Node $i$ (or the player that it represents in the rest of applications) with $N_i=\varnothing$ is called a \emph{source}. Negation exchanges $\lst$ and $H$ and
preserves $M$. Disjunction returns $H$ if some input is $H$,
$\lst$ if both inputs are $\lst$, and $M$ otherwise;
conjunction has the dual rule. Each node refines at most once.
These are restrictions of the update rules of
Algorithm~\ref{alg:propagation}. Some applications below decode the
final states into values different from $\nu$; their correctness is
proved from the state rules, rather than from the numerical
$1/6$-\GC guarantee.

\subsection{Divisible public goods games, preference games with strict  preferences, and financial networks}
\label{app:positive-goods-preferences-clearing}
This subsection treats three problems that share the same
best-response structure: divisible public goods with summation,
preference games with strict preferences, and clearing in restricted
CDS networks. After a change of variables, each asks for
$x\in[0,1]^V$ that is close to
$\max\{1-\sum_{j\in N_i}x_j,0\}$ coordinatewise
(see~\eqref{eq:strong-unit-nor} and~\eqref{eq:app-clearing-complement}).
For the two games, an elementary baseline already achieves error $1/2$:
give every source provision one and every other provision $1/2$, and
use Lemma~\ref{lem:strong-app-rows} to complete the preference rows. We are not
aware of a better general bound in~\cite{AppPP21,AppKPRST13}. For
financial networks, the corresponding trivial bound is the $1/2$
noted in~\cite{AppDHK25}.

We give two improvements. First, an error mapping from \eps-\GC
(Theorem~\ref{thm:strong-app-arithmetic}) gives error
$7\eps/(3+\eps)$ at in-degree two, which is $7/19$ at $\eps=1/6$.
This application genuinely needs the numerical threshold of
Theorem~\ref{thm:main}: the error is below $1/2$ only when
$\eps<3/13$, so a solver with an unspecified error below $1/2$ would
not suffice. Second, direct state propagation
(Theorem~\ref{thm:strong-app-direct-nor}) gives the stronger bounds
$1/5$ at in-degree two and $1/3$ without a degree bound. For
financial networks the two results differ in strength. The $1/5$
guarantee satisfies the $\eps$-approximate clearing notion of~\cite{AppDHK25},
including the solvency condition~\eqref{eq:strong-clearing-solvent},
whereas the $7/19$ guarantee gives only residual clearing. Both apply
to the restricted networks of Definition~\ref{def:app-unit-cds}, not
to arbitrary networks with CDSs.

\subsubsection{Definitions}

\subparagraph{Divisible public goods.}
We use the summation model and the best-response-distance notion of
approximation in Papadimitriou and Peng~\cite[Section~5.1]{AppPP21}.
The normalization below specifies the subclass to which our numerical
bounds apply.

\begin{definition}[Divisible public goods with summation]
\label{def:app-divisible}
An instance consists of a loop-free directed graph $G=(V,E)$.
For each player $i\in V$, let
$N_i=\{j:(j,i)\in E\}$ and $N(i)=N_i\cup\{i\}$.
Player $i$ chooses a provision $x_i\in[0,\infty)$ and receives utility
\[
 U_i(x)=b\!\left(\sum_{j\in N(i)}x_j\right)-c x_i,
\]
where $c>0$ is the common price per unit and $b$ is an increasing,
strictly concave valuation function. We normalize the standalone
optimum to one: $b(t)-ct$ has its unique maximum over
$t\in[0,\infty)$ at $t=1$. A profile $x\in[0,1]^V$ is an
\emph{$\eps$-approximate pure Nash equilibrium} if every provision
is within $\eps$ of its best response, that is,
\begin{equation}\label{eq:strong-unit-nor}
 \begin{gathered}
 T_i(x)=\max\left\{1-\sum_{j\in N_i}x_j,0\right\},\qquad
 \left|x_i-T_i(x)\right|\le\eps\quad(i\in V).
 \end{gathered}
\end{equation}
\end{definition}

Indeed, with $S=\sum_{j\in N_i}x_j$, maximizing player $i$'s
utility is equivalent to maximizing $b(t)-ct$ over $t\ge S$.
Strict concavity and the normalization give the stated best response.
Thus~\eqref{eq:strong-unit-nor} measures distance to a best response,
as in~\cite{AppPP21}, rather than utility regret. In particular, it is
not an $\eps$-NE in the sense of Definition~\ref{def:app-ne}. The algorithm needs
only the graph and this normalization. A different positive common
standalone optimum can be rescaled to one, with provision errors
rescaled accordingly.

\subparagraph{Preference games.}
We use the weight-distribution terminology and approximation convention
of Kintali et al.~\cite[Sections~3.1 and~3.5]{AppKPRST13}, restricted
to strict preference lists.

\begin{definition}[Preference game with strict preferences]
\label{def:app-preference}
Each player $i\in V$ has a strict total order $\succ_i$ on $V$,
including herself, and assigns weights $w_i(j)\ge0$ to the players.
Feasibility requires
\[
 \sum_{j\in V}w_i(j)=1,
 \qquad w_i(j)\le w_j(j)\quad(i,j\in V).
\]
An exact equilibrium is a feasible assignment in which every row is
lexicographically maximal with respect to the other rows: for every
feasible alternative $\widetilde w_i$ and every $j$,
\[
 \sum_{\ell:\,\ell\succeq_i j}w_i(\ell)
 \ge\sum_{\ell:\,\ell\succeq_i j}\widetilde w_i(\ell).
\]
For $\eps\ge0$, an \emph{$\eps$-equilibrium} is a weight
assignment satisfying, for all $i,j\in V$,
\begin{equation}\label{eq:strong-preference}
 \begin{gathered}
  \sum_{j\in V}w_i(j)=1,\qquad w_i(j)\ge0,
  \qquad w_i(j)\le w_j(j)+\eps, \quad \text{and} \\
  \sum_{\ell:\,\ell\succeq_i j}w_i(\ell)\ge1-\eps
  \quad\text{or}\quad |w_i(j)-w_j(j)|\le\eps.
 \end{gathered}
\end{equation}
Here $\succeq_i$ includes equality. Set
$N_i=\{j:j\succ_i i\}$ and $x_i=w_i(i)$.
The in-degree restriction counts players strictly preferred to self.
\end{definition}

In particular, the approximation convention permits capacity excess
up to $\eps$; exact capacity feasibility is required only when
$\eps=0$.

\begin{lemma}[Completion of preference rows]\label{lem:strong-app-rows}
Suppose $x\in[0,1]^n$ satisfies
\[
 \left|x_i-\max\left\{1-\sum_{j\in N_i}x_j,0\right\}\right|
 \le\eps\qquad\text{for every $i$},
\]
and $x_i=1$ whenever $N_i=\varnothing$. In polynomial time we can
construct a weight assignment with the prescribed diagonal entries
$w_i(i)=x_i$ satisfying~\eqref{eq:strong-preference}.
\end{lemma}

\begin{proof}
Fix $i$, put $S=\sum_{j\in N_i}x_j$ and $B=1-x_i$, assign
$w_i(i)=x_i$, and put zero weight below $i$ in her preference order.

If $S\ge B$, greedily fill capacities of players above $i$ in strict
preference order until their allocated sum is $B$. Every weight is at
most its capacity, and all entries before truncation equal their
capacities. If $x_i\le\eps$, the prefix at truncation and every
subsequent prefix above $i$ has weight $B\ge1-\eps$. If
$x_i>\eps$, the residual inequality implies $S<1$, and
\[
 S-B=x_i-(1-S)\le\eps.
\]
The total unallocated above-self capacity is at most $\eps$, so
each above-self weight differs from its capacity by at most $\eps$.
This also covers $B=0$.

If $S<B$, then $S<1$ and
\[
 \tau=B-S=(1-S)-x_i\le\eps.
\]
Here $N_i\ne\varnothing$, for otherwise the source assumption would
give $B=0$. Assign every player above $i$ her full capacity and add the
remaining $\tau$ to the least preferred above-self player. Every capacity
excess and every capacity difference is at most $\eps$. In both cases
the row sums to one, the diagonal entry equals its capacity, and the
prefix at or below self is one. All conditions in
\eqref{eq:strong-preference} hold, and the greedy construction is
polynomial-time.
\end{proof}

\subparagraph{Clearing in financial networks.}
We use the model and notation of
Dohn, Hansen and Klinkby~\cite[Section~2.2 and Definition~3]{AppDHK25}.

\begin{definition}[Clearing recovery rates]\label{def:app-clearing}
A financial network consists of banks $B$, external assets $e_i\ge0$,
debt notionals $c_{i,j}\ge0$ for distinct $i,j$, and CDS notionals
$c_{i,j}^k\ge0$.
In a CDS, $i$ is the writer, $j$ the holder, and $k$ the reference
bank; these banks are distinct. For a recovery-rate vector
$r\in[0,1]^B$, define liabilities, payments, and assets by
\[
 \ell_{i,j}(r)=c_{i,j}+\sum_k c_{i,j}^k(1-r_k),\qquad
 \ell_i(r)=\sum_j\ell_{i,j}(r).
\]
\[
 p_{i,j}(r)=r_i\ell_{i,j}(r),\qquad
 a_i(r)=e_i+\sum_jp_{j,i}(r).
\]
We assume no default costs and non-degeneracy: each bank has positive
external assets or writes positive unconditional debt. The clearing
map is
\[
 F_i(r)=
 \begin{cases}
  1,&a_i(r)\ge\ell_i(r),\\
  a_i(r)/\ell_i(r),&a_i(r)<\ell_i(r).
 \end{cases}
\]
A vector is \emph{clearing} if $r=F(r)$. We call it
\emph{$\eps$-residual clearing} if
$|r_i-F_i(r)|\le\eps$ for every bank. Under the approximation definition
of~\cite{AppDHK25}, an \emph{$\eps$-approximate clearing recovery
rate vector} must additionally satisfy
\begin{equation}\label{eq:strong-clearing-solvent}
 a_i(r)\ge(1+\eps)\ell_i(r)
 \quad\Longrightarrow\quad r_i=1.
\end{equation}
\end{definition}

We now give the definition of the restricted CDS networks used here.

\begin{definition}[Restricted CDS networks]
\label{def:app-unit-cds}
There are core banks $V$, CDS writers outside $V$, and an outside
creditor. Each core bank $i$ has $e_i=0$, writes only an unconditional
debt $c_{i,o}=1$ to the creditor $o$, and receives only CDS payments.
Every received CDS has notional one. Its reference is a core bank
other than its holder. Each writer $s$ is \emph{funded}: its external assets
cover its maximum possible liabilities, i.e., 
\[
 e_s\ge\sum_jc_{s,j}+\sum_{j,k}c_{s,j}^k>0.
\]
The creditor writes no contracts and has positive external assets.
All notionals and endowments are rational.
\end{definition}

Assign recovery one to the writers and creditor. Since each writer’s external assets cover its largest possible liabilities and the creditor has no liabilities, recovery one gives $r_i=F_i(r)$ for them, and also satisfies implication~\eqref{eq:strong-clearing-solvent}, whatever the recovery rates of the core banks. If $N_i$ is the list of references of the CDSs held
by core bank $i$, then
\begin{equation}\label{eq:strong-cds-only}
 \ell_i(r)=1,\qquad
 a_i(r)=\sum_{j\in N_i}(1-r_j),\qquad
 F_i(r)=\min\{1,a_i(r)\}.
\end{equation}
The list may contain a reference more than once; $|N_i|$ counts CDSs,
with multiplicity. Equivalently, contracts with a common reference
have their notionals added. Contracts referencing a bank already
fixed at recovery one contribute zero and can be removed beforehand.
With $x_i=1-r_i$, we obtain
\begin{equation}\label{eq:app-clearing-complement}
 1-F_i(r)=\max\left\{1-\sum_{j\in N_i}x_j,0\right\}.
\end{equation}
Thus complementation preserves the additive approximation in
\eqref{eq:strong-unit-nor}. Every core bank writes positive
unconditional debt, so the zero endowments respect non-degeneracy.

\subsubsection{Techniques and results}

\subparagraph{Error mapping.}
\begin{theorem}[Error mapping]\label{thm:strong-app-arithmetic}
For every fixed $0\le\eps<1/2$, a polynomial-time solver for
$\eps$-\GC gives a polynomial-time algorithm for
\eqref{eq:strong-unit-nor} with at most two inputs per coordinate and
additive error
\[
 e(\eps)=\frac{7\eps}{3+\eps}.
\]
Coordinates with no predecessors are assigned \(x_i=1\), which is exact (for restricted CDS networks, this means that a core bank holding no CDS is assigned recovery rate \(r_i=0\), consistent with its zero assets). In particular, Theorem~\ref{thm:main}
gives $7/19$-approximate pure Nash equilibria for the summation
public-goods model of Definition~\ref{def:app-divisible},
$7/19$-equilibria for preference games with strict preferences and
at most two players above self, and $7/19$-residual clearing for restricted CDS networks (Definition~\ref{def:app-unit-cds}) with at most two held CDSs  per
core bank.
\end{theorem}

\begin{proof}
Create the shared vertex $c=G_1()$. A vertex with no predecessors
uses the alias $v_i=c$. With one predecessor define $v_i=G_-(c,v_j)$;
with two define $t_i=G_+(v_j,v_k)$ and $v_i=G_-(c,t_i)$.
The circuit has linear size and each variable has a unique defining
gate. Every $\eps$-solution has $1-\eps\le c\le1$. Set
\[
 a=\frac3{4-c},\qquad b=\frac{1-c}{4-c},\qquad
 x_i=\begin{cases}
  1,&N_i=\varnothing,\\ b+av_i,&N_i\ne\varnothing.
 \end{cases}
\]
Since $a+b=1$, every $x_i$ lies in $[0,1]$.
For a coordinate with $k\in\{1,2\}$ inputs, let
$B_i=\max\{c-\sum_{j\in N_i}v_j,0\}$. As $c\le1$,
\[
 \max\{c-\min\{v_j+v_k,1\},0\}
 =\max\{c-v_j-v_k,0\}.
\]
Nonexpansiveness of clipping gives $|v_i-B_i|\le k\eps$.

First suppose none of the inputs is a source. Put
$d_k=(1-kb)/a-c$. Since $ac=1-4b$,
\[
 ad_1=3b,\qquad ad_2=2b,
\]
so $d_k\ge0$. The actual target is
\[
 T_i(x)=a\max\left\{c+d_k-\sum_{j\in N_i}v_j,0\right\}.
\]
The two truncated expressions with intercepts $c+d_k$ and $c$
differ by a value in $[0,d_k]$. Consequently,
\[
 b-ad_k-ak\eps\le x_i-T_i(x)\le b+ak\eps.
\]
For $k=2$, the absolute error is at most $b+2a\eps$; for
$k=1$, it is at most $2b+a\eps$, which is no larger since
$b/a=(1-c)/3\le\eps/3$. If some input is a source, then
$T_i(x)=B_i=0$, and $x_i\le b+ak\eps$. Sources themselves are
exact. Thus,
\[
 b+2a\eps=\frac{1-c+6\eps}{4-c}
 \le\frac{7\eps}{3+\eps},
\]
where the expression decreases with $c$ for $\eps<1/2$ and
$c\ge1-\eps$.

For preferences, apply Lemma~\ref{lem:strong-app-rows} to complete the
rows. For the core banks of Definition~\ref{def:app-unit-cds}, put $x_i=1-r_i$. Then
\[
 1-F_i(r)=\max\left\{1-\sum_{j\in N_i}x_j,0\right\},
\]
and complementation preserves the additive error. Funded writers are
substituted exactly.
\end{proof}

\subparagraph{Direct state propagation.}

\begin{lemma}\label{lem:strong-app-supported}
Run Algorithm~\ref{alg:propagation} on the following NOR circuit.
For each primary variable $v_i$, use $G_1$ if $N_i$ is empty,
$G_\neg(v_j)$ if $N_i$ has exactly one entry $j$, and a binary tree of $G_\vee$
gates followed by $G_\neg$ otherwise. The tree computes the OR
of the primary inputs in $N_i$; repeated references may be repeated
inputs. Every node has a unique producer. At termination, the states of its primary variables satisfy
\[
 \begin{aligned}
  \sigma(v_i)=\lst&\ \Longrightarrow\ \exists j\in N_i:\ \sigma(v_j)=H,\\
  \sigma(v_i)=H&\ \Longrightarrow\ \forall j\in N_i:\ \sigma(v_j)=\lst,\\
  \sigma(v_i)=M&\ \Longrightarrow\ \text{no predecessor is $H$ and some predecessor is $M$}.
 \end{aligned}
\]
Sources have state $H$. The construction permits arbitrary in-degree.
In particular, the low state is $\lst$, rather than the refinable state $L$.
\end{lemma}

\begin{proof}
The invariant state subset is $\{M,\lst,H\}$, with $\lst,H$
terminal. On this subset $G_\neg$ swaps the terminal states and fixes
$M$, while $G_\vee$ returns $H$ if an input is $H$, returns $\lst$ if both
inputs are $\lst$, and returns $M$ otherwise. At termination every
gate has the state prescribed by its rule. Induction through each
binary OR tree proves the implications. Each node refines at most
once, so Algorithm~\ref{alg:propagation} runs in polynomial time. No
unsupported terminal state is introduced.
\end{proof}

\begin{theorem}\label{thm:strong-app-direct-nor}
A direct run of
Algorithm~\ref{alg:propagation} yields a $1/5$-approximate pure Nash equilibrium for
Definition~\ref{def:app-divisible} with maximum in-degree two, a $1/5$-equilibrium for
Definition~\ref{def:app-preference} with at most two players above
self, and a $1/5$-approximate clearing recovery rate vector for
Definition~\ref{def:app-unit-cds} with at most two held CDSs per core bank. Without the degree restrictions, it gives
$1/3$-approximate pure Nash equilibria for summation public goods
and $1/3$-equilibria for preference games with strict preferences.
\end{theorem}

\begin{proof}
For in-degree at most two, decode states $\lst,H,M$ as $0,1,2/5$.
A vertex in state $\lst$ has a predecessor of value one and thus
$T_i(x)=0$; a vertex in state $H$ has all predecessors zero and
$T_i(x)=1$. These are exact. An unresolved vertex has one or two
predecessors of value $2/5$, so its target is $3/5$ or $1/5$,
both at distance $1/5$ from $2/5$. Sources receive exactly one.
Lemma~\ref{lem:strong-app-rows} then completes the preference rows.

For Definition~\ref{def:app-unit-cds}, put $r_i=1-x_i$, so residual errors are
preserved. If $x_i=0$, then $r_i=1$; if $x_i=1$, then assets are zero;
and if $x_i=2/5$, then assets are $2/5$ or $4/5$, strictly below
$1+1/5$. Hence the strengthened solvency implication holds.
Repeated references are counted with multiplicity in the sum and can
be represented by repeated inputs of an OR gate.

At arbitrary in-degree, decode $M$ as $1/3$. For an unresolved
vertex, with $k\ge1$ unresolved predecessors, the target is
$2/3$, $1/3$, or $0$ according as $k=1$, $k=2$, or $k\ge3$.
Each is within $1/3$ of $1/3$. The terminal-state argument and
row-completion lemma still apply. The clearing assertion uses the restriction to two held CDSs to verify 
\eqref{eq:strong-clearing-solvent}.
\end{proof}

\subsection{Binary best-shot public goods games}\label{app:positive-binary}
We consider the binary best-shot public goods game of Do Dinh and
Hollender~\cite{AppDDH24}, with common price $p\in(0,1)$, and give
guarantees for both well-supported and ordinary Nash equilibria. For
well-supported equilibria, the previous general guarantee is
$\min\{p,1-p\}$~\cite[Lemma~2.2]{AppDDH24}. Their matching hardness
result concerns graphs of unrestricted in-degree. We show that
\emph{any} polynomial-time \eps-\GC solver with $\eps<1/2$ gives
an improved guarantee $e_d(p)$ at every finite in-degree bound $d$
(Theorem~\ref{thm:strong-app-binary-wsne}). This gives exact
equilibria at $d=1$, and at $d=2$, $p=1/2$ it lowers the error from
$1/2$ to $(\sqrt5-2)/2\approx0.119$. For ordinary Nash equilibria,
the all-$1/2$ profile is an elementary $1/4$-NE at $p=1/2$ on every
directed graph, since then $|\Delta_i|\le1/2$ and each player mixes
uniformly. Direct state propagation lowers this to $\approx0.058$ at
in-degree two, and achieves $4/27$ for arbitrary in-degree and price
(Theorem~\ref{thm:strong-app-binary-ne}). Unlike the WSNE result,
the NE guarantee relies on the support properties of
Lemma~\ref{lem:strong-app-supported}, not merely on the existence of
an approximate circuit solution.

\begin{definition}[Indivisible public goods with the best-shot rule]
\label{def:app-binary}
An instance, as in~\cite[Definition~2]{AppDDH24}, consists of a
loop-free directed graph $G=(V,E)$ and a rational common price
$p\in(0,1)$. Each player $i$ chooses an action $a_i\in\{0,1\}$, where 1 means
producing the good, and 0 means not producing it. With $N_i=\{j:(j,i)\in E\}$, her payoff is
\[
 u_i(a_i,a_{-i})=
 \begin{cases}
  1-p,&a_i=1,\\
  1,&a_i=0\text{ and some }j\in N_i\text{ has }a_j=1,\\
  0,&\text{otherwise}.
 \end{cases}
\]
Write $s_i\in[0,1]$ for the probability of producing. The
approximation notions are those of Definition~\ref{def:app-ne}.
\end{definition}

The payoff difference between producing and abstaining is
\begin{equation}\label{eq:strong-binary-difference}
 \Delta_i(s)=\prod_{j\in N_i}(1-s_j)-p,
\end{equation}
where an empty product is one. Accordingly, a Nash equilibrium
approximation bounds
\[
 \operatorname{Reg}_i(s)=
 \max\{(1-s_i)\Delta_i(s),-s_i\Delta_i(s),0\},
\]
whereas a well-supported Nash equilibrium approximation bounds the payoff loss of each action assigned positive probability.

\subparagraph{Well-supported Nash equilibria (WSNE).}

\begin{theorem}\label{thm:strong-app-binary-wsne}
Consider the binary best-shot public-goods game of
Definition~\ref{def:app-binary}, with price $p\in(0,1)$
and maximum in-degree $d\ge1$.
Fix $0\le\eps<1/2$ such that $\eps$-\GC has a
polynomial-time solver. Then a polynomial-time algorithm
computes an $e_d(p)$-WSNE, where
\begin{equation}\label{eq:strong-wsne-root}
 a+a^d=2p,\quad a\in(0,1),\qquad
 e_d(p)=a-p=p-a^d.
\end{equation}
The probabilities may be algebraic. For every $\eta>0$,
rational outputs achieve $e_d(p)+\eta$ in time
polynomial in the input and $\log(1/\eta)$.
\end{theorem}

\begin{proof}
Apply the construction of Lemma~\ref{lem:strong-app-supported} to the
graph $G$ of the game, with one primary variable $v_i$ for each player
$i$ and $N_i$ the set of her in-neighbors. Thus $v_i$ is the NOR of the
variables of her in-neighbors, mirroring the best response in the
best-shot game: produce exactly when no in-neighbor produces. The
circuit has size $O(|V|+|E|)$. We use an arbitrary $\eps$-\GC
solution, rather than the propagation algorithm and its supported
states. Decode a circuit value at most $\eps$ as \emph{low}, a value
at least $1-\eps$ as \emph{high}, and any other value as
\emph{unresolved}. Here low, high, and unresolved refer to this numerical
decoding, not to the information states of Algorithm~\ref{alg:propagation}.  The bands are disjoint because $\eps<1/2$, and the
Boolean gates imply
\begin{equation*}
 \begin{aligned}
  \text{some predecessor high}&\ \Longrightarrow\ \text{output low},\\
  \text{all predecessors low}&\ \Longrightarrow\ \text{output high}.
 \end{aligned}
\end{equation*}
These implications propagate through each OR tree.

Assign the probabilities $0,1,1-a$ to primary variables with low, high,
and unresolved values, respectively. If some predecessor is high, the output is low and
abstention is an exact best response. If all predecessors are low,
the output is high and production is an exact best response.
Otherwise there is no high predecessor and there are
$k\in\{1,\ldots,d\}$ unresolved predecessors. Irrespective of the
player's own value, her payoff difference is
\[
 \Delta_i=a^k-p\in[a^d-p,a-p]
             =[-e_d(p),e_d(p)].
\]
Every supported action is therefore an $e_d(p)$-best response.

The function $a+a^d$ is strictly increasing from zero to two, so
the root in~\eqref{eq:strong-wsne-root} is unique. It satisfies
$a\ge p\ge a^d$, strictly for $d>1$, and
$e_d(p)<\min\{p,1-p\}$ for finite $d>1$. For $d=1$, $a=p$
and the error is zero. Bisection gives a rational approximation
$\widetilde a$ with $|a-\widetilde a|\le\eta/d$. As
$|a^k-\widetilde a^k|\le k|a-\widetilde a|$, this adds at most
$\eta$ to the WSNE error. Rational powers and bisection have
polynomial bit complexity, also if $d$ is part of the input.
\end{proof}

\subparagraph{Nash equilibria (NE).}

\begin{theorem}\label{thm:strong-app-binary-ne}
For Definition~\ref{def:app-binary} with maximum in-degree $d\ge1$,
Algorithm~\ref{alg:propagation} computes a
$\gamma_d(p)$-NE, where
\begin{equation}\label{eq:strong-ne-root}
 a^2+(1-a)a^d=p,\qquad \gamma_d(p)=a(a-p).
\end{equation}
The equation has a unique root $a\in(0,1)$. In particular,
\[
 \gamma_2(1/2) \approx 0.05788.
\]
Without an in-degree bound, the algorithm computes a $p(1-\sqrt p)$-NE,
whose worst-price bound is $4/27$. For every $\eta>0$, rational outputs achieve the stated
errors plus $\eta$ in time polynomial in the input and
$\log(1/\eta)$.
\end{theorem}

\begin{proof}
Decode states $\lst,H,M$ as $0,1,1-a$. Players with
terminal states play exact best responses by
Lemma~\ref{lem:strong-app-supported}. A remaining player has
$k\ge1$ unresolved predecessors, and her regret is
\[
 \max\{a(a^k-p),\,(1-a)(p-a^k),\,0\}
 \le\max\{a(a-p),\,(1-a)(p-a^d)\}.
\]
Equation~\eqref{eq:strong-ne-root} equalizes the two upper bounds.
Its left-hand side is strictly increasing on $(0,1)$: for $d\ge2$
its derivative is at least $2a-a^{d-1}>0$, and for $d=1$ it
equals $a$. The root satisfies $p\le a$ and $a^d\le p$, since
$a^2+(1-a)a^d\le a$ and $a^2+(1-a)a^d\ge a^d$.

For unrestricted in-degree, choose $a=\sqrt p$. Positive regret is at
most $a(a-p)=p(1-\sqrt p)$, and negative regret is at most
$(1-a)p$, the same value. The maximum of $p(1-\sqrt p)$ on
$[0,1]$ occurs at $p=4/9$ and is $4/27$. 
The regret expressions are Lipschitz in $a$ with
constant at most $d+2$, where $d$ is the maximum
in-degree of the input graph. A rational approximation $\widetilde a$ satisfying
$|\widetilde a-a|\le\eta/(d+2)$ can be computed in polynomial time by binary search, and increases the regret by at most $\eta$.
Since $d\le |V|-1$, this requires polynomial
time even when no fixed in-degree bound is imposed.
\end{proof}

\subsection{Polymatrix and Bayesian games}
\label{app:positive-polymatrix-bayesian}

For two-action polymatrix games, Deligkas et
al.~\cite[Theorem~4.2]{PureCircuit} give a
$1/3$-WSNE algorithm, while a modification described
in~\cite{AppEASSS23} achieves a $1/5$-NE.
We improve the latter bound to $0.191$.
The same improvement extends to two-player Bayesian
games and pairwise-additive multiplayer Bayesian games
through the type-agent representation~\cite{AppDFSS17}.
The corresponding $1/5$ Bayesian guarantee also
follows from this representation, rather than from
a previously published Bayesian result.

For comparison, Deligkas et
al.~\cite[Theorem~4.4]{PureCircuit} prove
PPAD-hardness of computing an $\eps$-NE in
two-action polymatrix games for every fixed
$\eps<2\sqrt{73}-17\approx0.088$.

\begin{definition}[Polymatrix game]\label{def:app-polymatrix}
A polymatrix game consists of players $V$, finite action sets
$A_i$, and rational matrices $P^{ij}\in\mathbb{Q}^{A_i\times A_j}$
for ordered pairs $i\ne j$. Missing interactions have zero matrices.
At an action profile $a$, player $i$'s total payoff is
\[
 u_i(a)=\sum_{j\ne i}P^{ij}_{a_i,a_j}.
\]
The two matrices $P^{ij}$ and $P^{ji}$ give the respective payoffs
on their common interaction. This is the standard sum of pairwise
games~\cite{AppDFSS17,PureCircuit}. The game is \emph{normalized}
if $u_i(a)\in[0,1]$ for every $i$ and every action profile $a$.
This normalizes total payoffs, not each interaction separately.
\end{definition}

\subsubsection{Two-action polymatrix games}
\label{subsec:strong-polymatrix}

We first consider normalized games with $A_i=\{0,1\}$ and no
restriction on the interaction graph. Write $x_i\in[0,1]$ for player $i$'s probability of
playing action one, and define the pure-action payoff difference
\begin{equation}\label{eq:poly-diff}
\Delta_i(x_{-i})=u_i(1,x_{-i})-u_i(0,x_{-i})
=b_i+\sum_{j\ne i}a_{ij}x_j.
\end{equation}
The difference is affine in opponents' probabilities and satisfies
\begin{equation}\label{eq:poly-norm}
-1\le\Delta_i(x_{-i})\le1\qquad(x_{-i}\in[0,1]^{n-1}).
\end{equation}
For an arbitrary strategy profile, player $i$'s regret is
\begin{equation}\label{eq:poly-regret}
\operatorname{Reg}_i(x)=
\begin{cases}
(1-x_i)\Delta_i(x_{-i}),&\Delta_i(x_{-i})\ge0,\\
-x_i\Delta_i(x_{-i}),&\Delta_i(x_{-i})<0.
\end{cases}
\end{equation}

Our construction invokes a standard complexity-theoretic result.
\begin{lemma}[Monotone weighted thresholds~\cite{AppBW06}]
\label{lem:strong-bw}
For rational nonnegative weights and rational threshold, the associated
Boolean weighted threshold function has a monotone AND/OR circuit of
size polynomial in the number of input variables and the encoding
length of the weights. Such a circuit can be constructed in polynomial
time, as shown explicitly by Beimel and Weinreb~\cite[Theorem~2.5]{AppBW06}
using sorting, merging, and halving circuits. Its gates may be taken to
have fan-in two.
\end{lemma}

We evaluate these circuits using the information states $\lst,M,H$
of Algorithm~\ref{alg:propagation}. The next lemma establishes
the completion properties needed by the approximation proof;
a numerical circuit solution alone need not have these properties.

A gate is \emph{stable} if its state equals the state prescribed by its
update rule on the current states of its inputs. For circuit inputs with states
in $\{\lst,M,H\}$, a \emph{completion} assigns a Boolean value to each
such input with state $M$ (unresolved input), and the encoded value to each terminal (resolved) input, that is, if its state is $\lst$ it assigns $0$ and if its state is $H$ it assigns $1$. This means that if the circuit has $k$ unresolved inputs, it admits $2^k$ distinct completions. The following result relies on the monotonicity of the circuit and focuses on only two completions: the
\emph{lower} and \emph{upper} completions, which assign $0$ and $1$, respectively, to
every unresolved circuit input.

\begin{lemma}[Three-state monotone circuit evaluation]
\label{lem:strong-threshold-states}
Let $K$ be a finite acyclic circuit of $G_\wedge$, $G_\vee$, $G_=$ and
constant gates, with $R(0)=\lst$ and $R(1)=H$, computing a monotone
Boolean function $f$ of inputs with states in $\{\lst,M,H\}$, and suppose
every gate of $K$ is stable. Then the output of $K$ is $H$ if every
completion gives $f=1$, is $\lst$ if every completion gives $f=0$, and
is $M$ otherwise. The same holds when each input is a literal, i.e.,
a primary variable $y$ or the output of $G_\neg(y)$, provided no primary
variable occurs with both polarities; completions are then taken over
the primary variables.
\end{lemma}
\begin{proof}
By $(a,b) \in \{ 0, 1 \}^2$ we denote the four possible pairs of values of $f$, where $a$ and $b$ are its values after the lower and upper completion, respectively. Since $f$ is monotone, its values at the lower and upper completions are
the smallest and largest values of $f$ over all completions. The pair
$(1,0)$ never arises, and on the pairs $(0,0),(0,1),(1,1)$ the
three-valued rules for $G_\wedge,G_\vee$ coincide with componentwise
minimum and maximum (e.g.\ $M\wedge M=M$ corresponds to
$(0,1)\wedge(0,1)=(0,1)$). Copy gates preserve these pairs, and constant
gates have their fixed values. Processing the gates in topological order,
stability gives each gate the pair of its values at the lower and upper
completions, and the pairs $(0,0),(0,1),(1,1)$ correspond to $\lst,M,H$.

For literal inputs, a negation of a primary variable interchanges its two
terminal states and preserves $M$. Since no primary variable occurs with
both polarities, completions of the primary variables correspond
bijectively to completions of the literals, a negated literal receiving
$1-b_y$, so the first part applies.
\end{proof}

The following theorem presents our improved results for polymatrix games by a direct application of the propagation technique of Algorithm~\ref{alg:propagation}. 

\begin{theorem}\label{thm:strong-polymatrix}
Every normalized rational two-action polymatrix game admits a
polynomial-time computable $0.191$-Nash equilibrium, with each
player's probability values belonging to
$\left\{0.191, 0.5, 0.809\right\}$.
More generally, every fixed rational error strictly above
$(3-\sqrt5)/4$ is achievable by the same method.
\end{theorem}
\begin{proof}
Fix a rational $p\in(0,1/2)$, and associate with each player a
primary Boolean-circuit variable $y_i$. In the intended decoding,
\begin{equation}\label{eq:poly-decode}
x_i=p+(1-2p)y_i,
\quad
\lst\mapsto p,\quad M\mapsto\tfrac12,\quad H\mapsto1-p.
\end{equation}
For every player $i$, form the Boolean function
\begin{equation*}
f_i(y_{-i})=
\mathbf1\left\{\Delta_i\bigl(p+(1-2p)y_{-i}\bigr)>0\right\}.
\end{equation*}
Since \eqref{eq:poly-diff} is affine, $f_i$ becomes a monotone
weighted threshold function by complementing precisely those inputs
$y_j$ whose coefficient $a_{ij}$ is negative. Its resulting weights
are nonnegative rationals of polynomial encoding length. Clearing
denominators turns the strict threshold into a non-strict integer
threshold (increase the integer right-hand side by one), so
Lemma~\ref{lem:strong-bw} applies.

Construct a monotone $G_\wedge/G_\vee$ circuit for each $f_i$, insert a negation
node for each complemented input literal, and identify its final output with the primary variable $y_i$. A constant threshold function is implemented by $G_0$ or $G_1$.
If a threshold circuit outputs an input literal directly, append
a copy gate $G_=$ before identifying its output with $y_i$.
The copy rule also preserves the three-state subset. Convert large-fan-in Boolean
gates into binary gates. Every circuit variable now has one producer,
and every individual threshold subcircuit is acyclic, although the
combined circuit may have directed cycles. The construction has
polynomial size. Choose state $\lst$ for constant zero and state $H$
for constant one, and run Algorithm~\ref{alg:propagation}. The three-state subsystem is invariant under these gates,
so Algorithm~\ref{alg:propagation} terminates in polynomial time with final states
$\sigma_i\in\{\lst,M,H\}$ for the primary variables.
Decode them by~\eqref{eq:poly-decode}.

At termination every gate is stable. Each threshold subcircuit is
acyclic, so Lemma~\ref{lem:strong-threshold-states} applies to it with
the final information states of its primary inputs, even though the
combined generalized circuit may contain cycles.

Fix a player $i$. By Lemma~\ref{lem:strong-threshold-states}, if
$\sigma_i=H$, then $\Delta_i>0$ at all completions of the
unresolved opponents. Such completions describe the corners of the
probability box in which resolved opponents play $p$ or $1-p$ and
unresolved opponents range over $[p,1-p]$. Since $\Delta_i$ is affine,
it is nonnegative throughout this box. Hence action one is a best
response at the decoded strategy profile, and by~\eqref{eq:poly-norm} and
\eqref{eq:poly-regret},
\[
\operatorname{Reg}_i(x)=(1-x_i)\Delta_i(x_{-i})\le p.
\]
Likewise, if $\sigma_i=\lst$, the threshold output is zero at every
completion, so $\Delta_i\le0$ throughout the box. Action zero is a
best response, and $\operatorname{Reg}_i(x)\le p$.

Suppose now that $\sigma_i=M$. The monotone threshold circuit
then has one completion with $\Delta_i\le0$ and another with
$\Delta_i>0$. Condition on the already fixed opponents, and let
$c_i$ be player $i$'s payoff difference when every remaining
unresolved opponent uses probability $1/2$. Let $s_i\ge0$ be the
half-range of this affine difference as the unresolved opponents range
over the full cube $[0,1]$. Its full-cube extrema are
$c_i-s_i$ and $c_i+s_i$. By~\eqref{eq:poly-norm},
\begin{equation}\label{eq:poly-center}
s_i+|c_i|\le1.
\end{equation}
Restricting the unresolved opponents to $[p,1-p]$ multiplies the
half-range by $1-2p$, while leaving the center unchanged. Since
$\Delta_i$ has opposite weak signs at two corners of that box,
\begin{equation}\label{eq:poly-cross}
|c_i|\le(1-2p)s_i
\le(1-2p)(1-|c_i|).
\end{equation}
Consequently,
\begin{equation}\label{eq:poly-unres}
|c_i|\le\frac{1-2p}{2(1-p)},
\qquad
\operatorname{Reg}_i(x)=\frac{|c_i|}{2}
\le\frac{1-2p}{4(1-p)},
\end{equation}
because the unresolved player herself mixes uniformly.
Combining the cases yields the uniform guarantee
\begin{equation}\label{eq:poly-eps}
\operatorname{Reg}_i(x)\le
\max\left\{p,\frac{1-2p}{4(1-p)}\right\}.
\end{equation}
The right-hand side of ~\eqref{eq:poly-eps} is minimized at
$p=(3-\sqrt5)/4 \approx 0.19099$. Choosing a fixed rational $p$ sufficiently close
to this point proves every fixed rational tolerance strictly above
that value.

All circuits have polynomial size by Lemma~\ref{lem:strong-bw},
Algorithm~\ref{alg:propagation} is polynomial, and the decoding uses
constant-bit rational probabilities for fixed $p$. This proves the
claimed running time.
\end{proof}

\subsubsection{Bayesian games}\label{subsec:strong-bayesian}

\begin{definition}[Bayesian game and interim equilibrium]
\label{def:app-bayesian}
A finite Bayesian game has players $1,\ldots,k$, explicitly represented
finite type sets $\Theta_i$, a rational common prior $\mu$ on type
profiles $t\in\prod_i\Theta_i$, and payoff functions $u_i(a;t)$.
Player $i$ observes only $t_i$ and chooses an action in
$A_i(t_i)$. A (mixed) strategy specifies a distribution
$s_{i,t_i}$ for each type. For every type of positive marginal
probability, define its interim pure-action payoff by
\[
 U_{i,t_i}(a;s_{-i})=
 \sum_{t_{-i}}\mu(t_{-i}\mid t_i)
 \mathbb{E}_{a_{-i}\sim\prod_{j\ne i}s_{j,t_j}}
 \bigl[u_i(a,a_{-i};t)\bigr].
\]
An \emph{$\eps$-interim Bayesian Nash equilibrium} satisfies
\[
 \max_{a\in A_i(t_i)}U_{i,t_i}(a;s_{-i})
 -\sum_{a\in A_i(t_i)}s_{i,t_i}(a)U_{i,t_i}(a;s_{-i})
 \le\eps
\]
for every such type. We normalize each type's total interim payoff
to $[0,1]$ for every own action and every opponent type-strategy
profile. No condition is imposed on zero-probability types.
\end{definition}

The type-agent representation for two-player games is the one used
in~\cite{AppDFSS17}. For two players we assume an explicit rational
table for the joint prior and realized payoffs. For the multiplayer
extension, payoffs are pairwise additive,
\[
 u_i(a;t)=\sum_{j\ne i}u_{ij}(a_i,a_j;t),
\]
and the conditional pairwise tables in
\eqref{eq:poly-bayes-multi} must be polynomial-time computable
rationals of polynomial encoding length. This assumption holds,
for example, when the prior and realized pairwise payoffs are
explicitly tabulated. Types may be correlated.

\begin{corollary}
\label{cor:strong-bayesian-two}
In a rational two-player Bayesian game with explicitly represented type
sets, an explicit rational joint prior and payoff table, two actions per
type, and normalized total interim payoffs as in
Definition~\ref{def:app-bayesian}, a polynomial-time algorithm
computes a $0.191$-interim Bayesian Nash equilibrium.
\end{corollary}
\begin{proof}
Ignore zero-marginal-probability types during computation, and assign
them arbitrary actions afterwards. With joint prior $\mu(t,s)$, create
one polymatrix player per positive-marginal type. For types $t$ of
player one and $s$ of player two, give their connecting edge the
respective payoff matrices
\begin{equation}\label{eq:poly-bayes-two}
A_{t,s}(a,b)=\frac{\mu(t,s)}{\mu_1(t)}u_1(a,b;t,s),\qquad
B_{t,s}(a,b)=\frac{\mu(t,s)}{\mu_2(s)}u_2(a,b;t,s).
\end{equation}
No two types of the same original player interact. The payoff to
clone $t$ from action $a$ is
\[
\sum_s\mu(s\mid t)\sum_{b\in\{0,1\}}x_{2,s}(b)
 u_1(a,b;t,s),
\]
which is exactly that type's conditional payoff in the original
Bayesian game. Hence the regret of each clone equals the conditional
regret of its corresponding type. Apply
Theorem~\ref{thm:strong-polymatrix}. The transformation has
polynomial size and rational bit complexity.
\end{proof}

\begin{corollary}
\label{cor:strong-bayesian-multi}
The same $0.191$ interim guarantee holds for rational multiplayer
Bayesian games with explicitly represented type sets, two actions per
type, and pairwise-additive
payoffs, provided all conditional pairwise payoff tables can be
computed in polynomial time and each type's total interim payoff is normalized as in
Definition~\ref{def:app-bayesian}.
\end{corollary}
\begin{proof}
Use the pairwise payoff decomposition in
Definition~\ref{def:app-bayesian}. Make one clone
$(i,t_i)$ per type of positive marginal probability. For distinct
original players $i,j$, define the edge payoff to clone $(i,t_i)$
from clone $(j,t_j)$ by
\begin{equation}\label{eq:poly-bayes-multi}
A_{(i,t_i),(j,t_j)}(a,b)=
\sum_{t_{-(i,j)}}
\mu(t_j,t_{-(i,j)}\mid t_i)
 u_{ij}(a,b;t_i,t_j,t_{-(i,j)}).
\end{equation}
Conditional expectation shows that the sum of the clone's edge
payoffs against opponent type strategies is exactly the original
type's interim payoff. This identity also holds for correlated
private types; no independence of types is assumed. Thus every
clone's regret equals the corresponding type's interim regret.
Apply Theorem~\ref{thm:strong-polymatrix}. As stated, the input
representation makes the conditional tables polynomial-time
computable with rational polynomial bit length.
\end{proof}

\subsection{Threshold games}\label{app:positive-threshold}

The circuit uses only three of the information states of
Algorithm~\ref{alg:propagation}; other gate types and the remaining
states are not needed. Unlike the previous applications, the decoding
here uses their original numerical values $\nu$. The guarantee
matches the algorithm and hardness threshold of~\cite{PureCircuit}.

\begin{definition}[Approximate threshold equilibrium]
\label{def:app-threshold}
A threshold game~\cite{AppPP21,PureCircuit} consists of a loop-free
directed graph $G=(V,E)$, with $N_i=\{j:(j,i)\in E\}$, and a
strategy $x_i\in[0,1]$ for each player. A strategy profile is an
\emph{$\eps$-approximate equilibrium} if, for every $i$,
\[
 \sum_{j\in N_i}x_j<\tfrac12-\eps
 \quad\Longrightarrow\quad x_i\ge1-\eps,
\]
\[
 \sum_{j\in N_i}x_j>\tfrac12+\eps
 \quad\Longrightarrow\quad x_i\le\eps.
\]
No constraint is imposed when the predecessor sum lies in the
closed interval $[1/2-\eps,1/2+\eps]$. These are the strict
premises of the $\eps$-\textsc{ThresholdGameNE} problem
in~\cite{PureCircuit}.
\end{definition}

\begin{theorem}\label{thm:strong-app-threshold}
Algorithm~\ref{alg:propagation} computes a $1/6$ threshold equilibrium
in polynomial time, with no restriction on in-degree. This matches
the tight bound of~\cite[Theorems~4.22--4.23]{PureCircuit}.
\end{theorem}

\begin{proof}
For each threshold vertex, use a $G_1$ gate at in-degree zero,
$G_\neg$ on the unique predecessor at in-degree one, and a $G_0$
gate at in-degree at least two. For the latter, select the permitted
terminal low state $\lst$. Run Algorithm~\ref{alg:propagation}, using
its original values $\nu(M)=1/3$, $\nu(\lst)=1/6$, and
$\nu(H)=5/6$. The three states form an invariant subposet, and all
values are at least $1/6$.

When there are at least two predecessors, the sum of their values
is at least $1/3$; the low-sum premise cannot hold, and the constant
output $1/6$ satisfies every high-sum requirement. Sources have
value $5/6$. With exactly one predecessor, the possible inputs are
$1/6$, $1/3$, and $5/6$. The first forces the output $5/6$,
the last forces $1/6$, and the middle forces neither condition.
The negation rule leaves state $M$ unresolved. Thus every threshold
gate is satisfied, including those incoming edges that were not
represented in the circuit for high-in-degree vertices.
\end{proof}

\subsection{Violating equilibrium conditions for a fraction of players}
\label{app:positive-fragment}

Our decomposition theorem gives polynomial-time algorithms
that satisfy the equilibrium conditions of all but a small
fraction of players in constant-action polymatrix and
graphical games of maximum in-degree two. The guarantee is
$3/4-o(1)$ and does not rely on the numerical $1/6$-\GC
bound or on state propagation. Arieli and
Babichenko~\cite{AppAB17} study weak approximate equilibria
in explicitly represented normal-form games, but their
polynomial-time guarantees do not directly extend to the
succinct representations considered here.

\begin{definition}[Graphical game and dependency digraph]
\label{def:app-graphical}
A graphical game specifies finite action sets $A_i$ and, for
each player $i$, a set $N_i\subseteq V\setminus\{i\}$ and a local
payoff table
\[
 u_i:A_i\times\prod_{j\in N_i}A_j\longrightarrow\mathbb{Q}.
\]
The dependency digraph has an arc $j\to i$ for each $j\in N_i$:
arcs point from an opponent whose action can affect $i$'s payoff
towards $i$. For a polymatrix game, put $j\in N_i$ if some row of
$P^{ij}$ is nonconstant across $j$'s actions. Thus $N_i$ contains
precisely the opponents whose actions can affect $i$'s total payoff.
Contributions from the other matrices depend only on $i$'s own
action. The graph is simple and loop-free.
Write $n=|V|$ and $n_2=|\{i:|N_i|=2\}|$.
\end{definition}

We say that a player \emph{satisfies her $\eps$-WSNE conditions}
at a strategy profile if all the actions she plays with positive probability in her strategy satisfy the inequalities
in Definition~\ref{def:app-ne}; the other players need not satisfy
their conditions. The following statement concerns strategy profiles with
this property at many players, rather than a WSNE of the entire game.

\begin{theorem}
\label{thm:strong-app-fraction}
For a rational polymatrix game with $|A_i|\le s$ for a fixed
constant $s$, and a dependency digraph of maximum in-degree two, a
polynomial-time algorithm returns a rational strategy profile in which at
least
\begin{equation}\label{eq:strong-app-players-count}
 n-\frac{n_2}{4}-\frac{4n}{\lambda+1},
 \qquad \lambda=\max\{1,\lfloor\log_2(n+1)\rfloor\},
\end{equation}
players have every supported action as an exact best response.
Here $n_2$ is the number of vertices of dependency in-degree
two. For graphical games with the same restrictions, payoff values in
$[0,1]$, and fixed $\eps>0$, the same count satisfies their $\eps$-WSNE, hence
$\eps$-NE, conditions. Thus the loss beyond $n_2/4$ is
$O(n/\log n)$, and the worst-case satisfied fraction is
$3/4-o(1)$.
\end{theorem}

\begin{proof}
Apply Theorem~\ref{thm:fragmentation} to the dependency
digraph. It computes a set $S$ with
\[
 |S|\le n_2/4+4n/(\lambda+1)
\]
such that every strongly connected component of the remaining
digraph has at most $\lambda$ players. Fix arbitrary pure
actions for the players in $S$, and process the components of
the remaining graph in topological order, directing arcs from
influencers to influenced players. When a component is reached,
all inputs from outside it have already been fixed.
\\

\noindent\emph{Polymatrix games.}
For each component, enumerate all choices of nonempty support
sets $T_i$ and solve the rational linear feasibility system
\begin{equation}\label{eq:strong-app-polymatrix-LP}
 \begin{gathered}
 p_{i,a}\ge0,\quad \sum_a p_{i,a}=1,\quad
 p_{i,a}=0\quad(a\notin T_i),\\
 U_{i,a}(p)=v_i\quad(a\in T_i),\qquad
 U_{i,a}(p)\le v_i\quad(a\notin T_i).
 \end{gathered}
\end{equation}
Pure-action payoffs $U_{i,a}$ are affine in the other players'
probabilities, so these are linear programs. Nash's theorem
ensures that the supports of some exact equilibrium give a
feasible system; conversely every feasible solution is an
exact equilibrium of the component. There are at most
$(2^s-1)^\lambda=n^{O(s)}$ systems. Each has
$O(s\lambda)$ variables with rational coefficients, so a
feasible rational vertex can be found in polynomial time.
The rational outputs retain polynomial bit complexity:
coefficient matrices have polynomial-size determinants, while
previously fixed probabilities occur only on right-hand sides.
After clearing denominators in the original matrices once, let $D$
be a common denominator for already fixed probabilities. Each
component's right-hand sides have a denominator dividing $D$
times the fixed input denominator. Cramer's rule for a basic
feasible solution adds only a determinant of an integer coefficient
matrix built from the original payoff coefficients and simplex
constraints. Its bit length is polynomial in the input, independently
of $D$. Over at most $n$ components, these added bit lengths sum
to a polynomial.
\\

\noindent\emph{Graphical games.}
Set
\[
 q=\max\{1,\lceil4(s-1)/\eps\rceil\}.
\]
Enumerate all strategy profiles with probabilities that are multiples
of $1/q$, giving at most
$(q+1)^{s\lambda}$ candidates per component. Some candidate
satisfies every component player's $\eps$-WSNE inequalities.
To see this, take an exact equilibrium of the component.
For each player round all but one positive-support probabilities
down to multiples of $1/q$ and add the removed mass to the
remaining positive-support action. The support can only shrink,
and the total variation of each distribution changes by at most
$(s-1)/q$. As payoffs lie in $[0,1]$, and each pure payoff depends on at
most two changing predecessors, its expected value changes by at most
$2(s-1)/q$. Every action retained in the support was an
exact best response before rounding; its payoff gap after
rounding is therefore at most
$4(s-1)/q\le\eps$. Exhaustive enumeration finds such a
profile by checking the inequalities.
For fixed $s,\eps$, the enumeration is polynomial-time.

Later components do not influence already solved ones.
Only the players of $S$ may fail their conditions, proving
\eqref{eq:strong-app-players-count}. Averaging each player's
WSNE support inequalities gives the corresponding NE condition.
\end{proof}
\end{document}